\documentclass[11pt, letter]{article}

\usepackage{stolenstyle}

\usepackage{amsmath,epsf,amssymb,latexsym,amsthm,setspace,bbm,array,pifont,enumerate}

\DeclareFontFamily{U}{rsf}{}
\DeclareFontShape{U}{rsf}{m}{n}{
  <5> <6> rsfs5 <7> <8> <9> rsfs7 <10-> rsfs10}{}
\DeclareMathAlphabet\Scr{U}{rsf}{m}{n}

\def\GUL{\GU(1)_{\text{L}}}
\def\GUR{\GU(1)_{\text{R}}}

\def\C{{\mathbb C}}

\def\Q{{\mathbb Q}}

\def\Z{{\mathbb Z}}

\def\deg{\operatorname{deg}}

\def\SO{\operatorname{SO}}

\def\GU{\operatorname{U{}}}

\def\p{\partial}

\def\la{\langle}
\def\ra{\rangle}

\def\ff#1#2{{\textstyle\frac{#1}{#2}}}

\def\cN{{\cal N}}

\newcommand\thetab{\overline{\theta}}

\newcommand\cb{\overline{c}}

\newcommand\zb{\overline{z}}

\newcommand\Zh{\widehat{Z}}

\newcommand\Xb{\overline{X}}

\newcommand\Wt{\widetilde{W}}

\newtheorem{thm}{Theorem}[section]
\newtheorem{lem}[thm]{Lemma}
\newtheorem{propo}[thm]{Proposition}

\theoremstyle{definition}
\newtheorem{defn}[thm]{Definition}

\usepackage{tikz}

\usetikzlibrary[shapes.geometric]
\usetikzlibrary{positioning} 
\usetikzlibrary{calc,intersections,through,backgrounds}

\tikzset{>=stealth}
\tikzset{every picture/.style={very thick}}

\def\bW{{{\boldsymbol{W}}}}

\def\bM{{\boldsymbol{{M}}}}

\def\qtot{{{q_{\text{tot}}}}}

\def\M#1{{{\frac{1}{m_{#1}}}}}

\def\preprint{ ~~ }

\title{Landau--Ginzburg skeletons}
\author {Ian C.~Davenport and Ilarion V.~Melnikov}
\affiliation{Department of Physics and Astronomy,\\
James Madison University, Harrisonburg, VA 22807, USA}
\emailAdd{davenpic@dukes.jmu.edu, melnikix@jmu.edu}

\abstract{We study the class of indecomposable two-dimensional Landau-Ginzburg theories with (2,2) supersymmetry and central charge $c < 6$ with the aim of classifying all such theories up to marginal deformations.  Our results include cases overlooked in previous classifications.  The results are rigorous for three or fewer fields and more generally are rigorous if we assume an extra bound.  Numerics suggest that we have the complete set of indecomposable Landau-Ginzburg families with $c<6$.  This set consists of $38$ infinite families and a finite list of $418$ sporadic cases.  The basic tools are classic results of Kreuzer and Skarke on quasi--homogeneous isolated singularities and solutions to certain
feasibility integer programming problems.  }

\begin{document}

\maketitle

\section{Introduction}\label{s:intro}
Landau-Ginzburg (LG) theories have played an important role in the study of critical phenomena.  This is especially true in the case of two dimensions, where they yield Lagrangian theories that flow to a wide class of interacting IR fixed points.  There has been important work on LG theories with various amounts of supersymmetry.  For instance, the seminal work of~\cite{Zamolodchikov:1986db} pointed out the connection between minimal models and LG theories; theories with (1,1) supersymmetry were discussed in~\cite{Kastor:1988ef}; theories with (2,2) supersymmetry were introduced in~\cite{Martinec:1988zu,Vafa:1988uu}.  The defining data of (2,2) LG theories is intimately related to the mathematics of isolated quasi--homogeneous hypersurface singularities, a subject studied for many years, going back to, for example, the classic work~\cite{Milnor:1970or}.  The graphical techniques of the sort we will be using go back to~\cite{Arnold:1985gz}.\footnote{A comparative history of these mathematical developments is presented in~\cite{Hertling:2012ku}, where also a number of interesting classification results are obtained, focusing on classifications at fixed Milnor number.}  

The purpose of this work is to describe the combinatorics of (2,2) LG theories and classify indecomposable theories with central charge $c < 6$.  Our interest in this question is threefold.  First, a rich correspondence between these LG theories and  $\cN = 2$, $d=4$ superconformal field theories was proposed in~\cite{DelZotto:2015rca}.  It has been investigated in a geometric setting of rational Gorenstein singularities in~\cite{Xie:2015rpa,Chen:2016bzh,Wang:2016yha} based on the earlier~\cite{Yau:2003sst}.  Our work generalizes those studies.  We find several isolated rational Gorenstein hypersurfaces singularities in $\C^4$ overlooked in~\cite{Yau:2003sst}, and, more generally, we also describe LG theories that do not correspond to hypersurface singularities in $\C^4$.  Second, the $c<6$ boundary is interesting since theories with $c<6$ can arise in descriptions of various singular limits of $(4,4)$ superconformal theories with $c=6$.  We have in mind here the kinds of limits described in~\cite{Giveon:1999zm}.  Finally, a good understanding of the (2,2) combinatorics should serve as a springboard for efforts to classify (0,2) LG theories.  Although these were introduced and studied many years ago in~\cite{Distler:1993mk,Kawai:1994qy}, we have a very incomplete understanding and no classification even at fixed central charge!

The tools we use were developed for the purpose of classifying (2,2) LG theories at fixed central charge;  $c=9$ theories were classified in~\cite{Kreuzer:2002uu,Kreuzer:1992np,Klemm:1992bx}, and $c=12$ theories in~\cite{Lynker:1998pb}.\footnote{These classification results are available at the Kreuzer Calabi-Yau database:  \url{http://hep.itp.tuwien.ac.at/~kreuzer/CY/}.}  We will review them, specifically the elegant approach developed in~\cite{Kreuzer:1992np}, in the next section and in section~\ref{s:links}, and then we will turn to the $c<6$ classification problem in sections~\ref{s:skeletons} and~\ref{s:pointing}.  We end with a summary of our findings.

\section*{Acknowledgements}  IVM would like to thank M. Del Zotto for sparking his interest in the question.  The work of ICD was supported in part through the Tickle Foundation, and the work of IVM was supported in part through the Faculty Assistance Grant from the College of Science and Mathematics at James Madison University.

\section{Quasi-homogeneous combinatorics} \label{s:combinatorics}
A (2,2) LG theory is a theory specified by a UV Lagrangian for $n$ chiral superfields $X_i$ with canonical quadratic kinetic terms and polynomial superpotential interactions:
\begin{align}
S = \int d^2 z d^2\theta d^2\thetab \sum_{i=1}^n \Xb_i X_i +  \left\{ \int d^2 z d\theta^+d\thetab^+~ W(X) + \text{h.c.} \right\}~.
\end{align}
Here $(z,\zb)$ are coordinates on the Euclidean world-sheet, and $\theta^\pm$, $\thetab^\pm$ are the corresponding superspace coordinates.  The $\theta^\pm$ ($\thetab^\pm$) carry $\GUL$ ($\GUR$) charges $\pm 1$, and $d^2z d\theta^+d \thetab^+$ is the chiral superspace measure.  We demand that the interactions preserve a $\GUL\times\GUR$ R-symmetry, which requires the superpotential $W(X)$ to be a quasi--homogeneous polynomial.  That is there are non-negative rational charges $q_i \in \Q$ such that for any $t \in \C^\ast$
\begin{align}
W(t_i^{q_i} X_i) = t W(X)~.
\end{align}
The theory has a supersymmetric RG flow that preserves the $\GUL\times\GUR$ symmetry and flows to a strongly coupled theory in the IR.   We will make the standard assumption that the $\GUL\times\GUR$ symmetries correspond to the left and right symmetries of a (2,2) SCFT.  In particular, the chiral field $X_i$ will flow to a chiral primary (more precisely (c,c) ) operator with $\GUL\times\GUR$ charge $(q_i,q_i)$.  Since the theory is left-right symmetric, in what follows we will simply refer to $q_i$ as the R-charge of $X_i$.

 We will also require that the theory has a normalizable vacuum.  We will assume this is equivalent to the scalar potential of the LG theory having an isolated vacuum, i.e. $W(X)$ is \textit{compact}.  
\begin{defn}
Let $W_i = \frac{\p W}{\p X_i}$.  We say that $W(X)$ is \textit{compact} if and only if $X_1 = X_2 = \cdots = X_n = 0$ is the unique solution to $\{W_1 = 0, W_2 = 0,\ldots, W_n = 0\} \in \C^n$; in other words, $W$ has a unique critical point.
\end{defn} 
The following are equivalent:
\begin{enumerate}
\item $W(X)$ is compact;
\item  the ideal $\bW = \la W_1, W_2,\ldots, W_n\ra \subset \C[X_1,X_2,\ldots, X_n]$  is zero-dimensional;
\item the Jacobian ring $R_{W} = \frac{\C[X_1,\ldots,X_n] }{ \bW}$ is finite dimensional as a vector space over $\C$;
\item the determinant of the Hessian matrix of $W(X)$ is non-zero in $R_W$.
\end{enumerate}
These equivalences are well-known; see, for instance,~\cite{Lerche:1989uy}.

With our assumptions the central charges of the (2,2) SCFT are determined by the anomalies of the UV $\GUL\times\GUR$ symmetries: they are equal and given by
\begin{align}
c_L = c_R = c = 3 \sum_{i=1}^n (1-2 q_i)~.
\end{align}
The classification of LG theories at a fixed $c$ amounts to finding all $n$ and charges $(q_1,\ldots,q_n) \in \Q^n$ that yield the desired $c$ and admit a compact superpotential.  While we will not use this result, we note~\cite{Kreuzer:1992np} that for any fixed rational $c$ there is a finite set of $(q_1,\ldots, q_n) \in \Q^n$ that yield compact theories. 

It may be that a compact $W$ is \textit{decomposable}, i.e. it can be written as
\begin{align}
W = W^{(1)}(X_1,\ldots, X_k) + W^{(2)}(X_{k+1},\ldots, X_{n})~.
\end{align}
In that case $W^{(1)}$ and $W^{(2)}$ both define compact LG theories with, respectively, $k$ and $n-k$ variables.  Since these are clearly non-interacting, the IR fixed point will be a product of the two SCFTs, and in our classification of LG families we may as well ignore such decomposable theories.\footnote{Note that a compact LG theory has additional global symmetry if and only if it is decomposable~\cite{Bertolini:2014ela}.}  In the graphical notation that we will introduce shortly this will amount to ignoring disconnected graphs.

\subsection{Constraints on the charges}
To describe the classification further, we will need to introduce some more notation and results.  Most of these are discussed in~\cite{Kreuzer:1992np}, and we include them here for completeness and reference.  Throughout, $W$ will be a  polynomial quasi--homogeneous superpotential, i.e. it takes the form
\begin{align}
W = \sum_{p \in \Delta \cap \Z^n} a_p \prod_{i=1}^n X_i^{p_i}~,
\end{align}
where the $a_p$ are complex coefficients, and $\Delta$ is the Newton polytope for $W$.  Because $W$ is polynomial and quasi-homogeneous, $\Delta$ lies in the positive orthant  intersected with the hyperplane $\sum_{i=1}^n q_i p_i = 1$ .  Thus, the charges $(q_1,\ldots,q_n)$ specify a family of LG theories: the coefficients $a_p$ modulo the action of holomorphic field redefinitions correspond to marginal deformations of the theory.  We will say that the family determined by $(q_1,\ldots,q_n)$ is \textit{generically compact} (or \textit{gc} for short) if $W$ is compact for a generic choice of the coefficients $a_p$.  For a gc family the basic properties of the SCFT, such as the Zamolodchikov metric and three-point functions vary smoothly with these parameters away from a complex co-dimension one singular locus where $W$ fails to be compact.

\begin{defn} For a fixed superpotential $W$  a variable $X_i$ is a \textit{root} if and only if $W$ contains the monomial  $X_i^{m_i}$.\footnote{In what follows we will use a short-hand for this:  $W \supset X_i^{m_i}$.}  In this case we call the integer $m_i$ the \textit{exponent of the root}.   A variable that is not a root is a \textit{pointer} if and only if  $W$ contains the monomial $X_i^{m_i} X_j$ with $j\neq i$; in this case we say  \textit{$X_i$ points at $X_j$} and has \textit{exponent} $m_i$.
\end{defn}
\begin{lem}
If $W$ is compact, then every variable is either a root or a pointer.  
\end{lem}
\begin{proof}  Suppose $W$ is a superpotential with variable $X_1$ that is neither a root or a pointer.  Every monomial in $dW$ contains at least one $X_j$ for $i>1$.  Hence, $dW|_{X_{i>1} = 0} = 0$ for all $X_1$, and the potential is not compact.
\end{proof}
\begin{defn} A superpotential $W$ is IR--equivalent to $\Wt$ if and only if the two LG theories flow to the same fixed point.
\end{defn}
There are two simple examples of IR equivalence:  $W$ and $\Wt$ may be related by   a change of variables compatible with quasi--homogeneity, or the two may be equivalent up to massive fields, in which case we can drop the massive fields without affecting the IR fixed point.

\begin{lem}
Up to IR equivalence we can assume $0 < q_i \le 1/2$.
\end{lem}
\begin{proof}
Let $Y$ be any variable with $q_Y >1/2$ and denote the remaining fields by $Z$. With our assumptions on the charges
\begin{align}
W &= Y F(Z) + G(Z) &  \implies&&
dW & = dY F(Z) + (Y dF + dG)~.
\end{align}
$F(Z)$ has degree $1-q_Y$.  Any quasi--homogeneous polynomial $F(Z)$ with positive charge vanishes at $Z=0$, so there are two possibilities:
\begin{enumerate}
\item $q_Y = 1$.  In this case $F(Z) = F_0$, a constant, and $W$ has no critical point.
\item $q_Y < 1$.  In this case $dW|_{Z=0} = Y dF|_{Z=0}$ for all $Y$.  Hence $W$ is singular unless among the $Z$ there is a field $\Zh$ with $q_{\Zh} = 1-q_Y$ and $F(Z) =  \Zh + \ldots$, where ``$\ldots$'' denote non-linear terms in the $Z$; these must be independent of $\Zh$.  In other words,
\begin{align}
W = Y (\Zh - F_{\text{nlin}} (Z) ) + G(Z,\Zh)~.
\end{align}
Hence, $Y$ is a Lagrange mutliplier enforcing $\Zh = F_{\text{nlin}}$.  Since the right-hand-side is $\Zh$ independent, we can solve the condition and obtain the equivalent superpotential
\begin{align}
W = Y\Zh + G(Z,F)~.
\end{align}
This is IR--equivalent to $W = G(Z,F)$, a potential without either the $Y$ or $\Zh$ field. 
\end{enumerate}
Repeating this for every field with $q_Y > 1/2$, we obtain an IR-equivalent potential where every field has $q\le 1/2$.
\end{proof}

In fact, we have a stronger result.
\begin{propo}
Up to IR equivalence we can assume $0<q < 1/2$ for every compact potential $W$; in particular, every monomial in $W$ is at least cubic in the fields.
\end{propo}

\begin{proof}
By the previous lemma, we already have $q_i \le 1/2$.  Now we decompose the fields into $Y_a$, $a=1,\ldots,k$, with $q_a = 1/2$ and $Z_s$ with $q_s < 1/2$.   After an $\SO(k)$ linear transformation on the $Y_a$, we find that the general form of the potential is
\begin{align}
W = \sum_{a=1}^k \left[ \ff{m_a}{2} Y_a^2 - Y_a F_a(Z) \right] + G(Z)~.
\end{align}
Hence,
\begin{align}
dW = \sum_{a=1}^k dY_a \left[ m_a Y_a - F_a(Z) \right] + dG -\sum_{a=1}^k Y_a dF_a~.
\end{align}
Since $\deg F_a=1/2$ and $q_Z < 1/2$, it follows that  $dW|_{Z=0} = \sum_a Y_a m_a dY_a$.  Hence, $W$ is non-compact if any $m_a = 0$.  On the other hand, if $m_a >0$ for all $a$, then we can use the equation of motion $Y_a = m_a^{-1} F_a$ to obtain the IR-equivalent potential
\begin{align}
W = G(Z) - \sum_{a=1}^k \frac{F_a(Z)^2}{2 m_a}~.
\end{align}
\end{proof}

\subsection{Constraints on $n$ from $c$}
A gc family leads to a family of IR fixed points with
\begin{align}
c = \cb = 3 \sum_{i=1}^n  \left( 1- 2q_i\right)~.
\end{align}
Since every $q_i >0$ it follows that $n > c/3$. We also see that as $q_i \to 0$, we recover the free field central charge values~\footnote{A class of intriguing theories have been obtained by taking large $q_i$ limits of minimal models in~\cite{Fredenhagen:2012rb,Fredenhagen:2012bw,Gaberdiel:2016xwo}; it may be interesting to extend that construction to more general LG theories, such as the infinite families constructed in this work.}; of course in the strict limit we would get a non-compact theory.  Here is a less trivial consequence, also proven in~\cite{Kreuzer:1992np}.
\begin{propo}  \label{propo:cbound} A gc family has $c/3<n \le c$.
\end{propo}
Hence, to describe theories with $c<6$ we consider $1\le n \le 5$ and
\begin{align}
\qtot = \sum_{i=1}^n q_i  > \frac{3n-6}{2}~.
\end{align}
In other words, we have
\begin{align}
\label{eq:qtotbounds}
n &< 3 & n & = 3 & n&=4 & n&=5  \nonumber \\
\qtot &>0 &\qtot &>\ff{1}{2}  &	\qtot&>1	& \qtot &>\ff{3}{2}
\end{align}

To prove the non-trivial part of the proposition, observe that every field with $q_i \le 1/3$ contributes at least $1$ to the central charge, and a pair of fields $X,Y$ with $1/3<q_X < 1/2$ and $q_Y = 1-2q_X$ contributes at least $2$ to the central charge.  The claim then follows from
\begin{lem}
Fix a gc family.  Suppose there are $k$ fields $Z_a$ with $q(Z_a) \in (1/3,1/2)$.  Then for every $Z_a$ there is a distinct field $Y_a$ with $q(Y_a) = 1-2q(Z_a)$. 
\end{lem}
\begin{proof}  $Z_a$ must be a root or a pointer, but its charges only allow it to be a pointer of the form $W \supset Z_a^2 Y_a$, with $q(Y_a) = 1-2 q(Z_a)$.  Suppose $Z_1$ and $Z_2$ point to the same field $Y$.  Then $Z_{1,2}$ have the same charge, and the potential is of the form 
\begin{align*}
W = Q_1(Z_1,Z_2) Y + Q_2(Z_1,Z_2) F(U) + G(Y,U)~,
\end{align*}
where $U$ denote all the other fields, and $Q_{1,2}$ are quadratic polynomials.  It is then easy to see that $dW = 0$ whenever $Y=U=0$ and $Q_1(Z_1,Z_2) = 0$, and hence $W$ is not compact.
\end{proof}
To finish the proof of the proposition, we group our fields by the pairs $(Z_a,Y_a)$, as well as any remaining fields $U$ with $q(U) \le 1/3$.

\subsection{A graphical notation and the Kreuzer-Skarke theorem}
Every compact potential $W$ with $n$ variables yields a directed graph, known as the \textit{skeleton graph} of $W$, obtained as follows.   Every variable corresponds to a node indicated by an open circle $\circ$;  there is an arrow from node $i$ to node $j$ if and only if $X_i$ points to $X_j$.  In general a $W$ may yield several distinct graphs.  For instance, we have the following potentials and corresponding graphs with $n\le 2$:
\begin{align}
W & = X_1^{m_1}~~&
\begin{tikzpicture}[baseline=-3mm]
\node [] (1)   [label=below:$m_1$] {$\circ$};\end{tikzpicture}~ 
\nonumber\\[2mm]
W  &\supset X_1^{m_1} + X_2^{m_2}~~&
\begin{tikzpicture}[baseline=-3mm]
\node [] (1)   [label=below:$m_1$] {$\circ$};
\node [] (2)  [right=of  1, label=below:$m_2$] {$\circ$};
\end{tikzpicture}~ \nonumber\\[2mm]
W &\supset  X_1^{m_1} + X_1 X_2^{m_2}~~&
\begin{tikzpicture}[baseline=-3mm]
\node [] (1)   [label=below:$m_1$] {$\circ$};
\node [] (2)  [right=of  1, label=below:$m_2$] {$\circ$};
\draw [->]  (2) -- (1);
\end{tikzpicture}~ \nonumber\\[2mm]
 W &\supset X_1^{m_1}X_2 + X_1 X_2^{m_2}~~&
 \begin{tikzpicture}[baseline=-3mm]
\node [] (1)   [label=below:$m_1$] {$\circ$};
\node [] (2)  [right=of  1, label=below:$m_2$] {$\circ$};
\draw [->]  (2) to [out = 160, in = 20] (1);
\draw[->] (1)  to [out = -20, in = -160] (2);
\end{tikzpicture}~
\end{align}
Note that we decorate each node with an exponent label; this is a convenient way of keeping track of the combinatoric and integer data in the same graphical notation.


Every gc family with $n=2$ has a potential of one of these forms.  Note that a skeleton graph uniquely determines the charges in terms of the exponents.

A skeleton potential, even if taken with generic coefficients, need not be smooth.  To see this, consider the $n=3$ skeletons:
\begin{align}
\ge 2~\text{roots} &&
\begin{tikzpicture}[baseline=-2mm]
\node [] (1)   [label=below:$m_1$] {$\circ$};
\node [] (2)  [right=of  1, label=below:$m_2$] {$\circ$};
\node [] (3) [right=of 2, label = below: $m_3$] {$\circ$};
\end{tikzpicture}~ &&
\begin{tikzpicture}[baseline=-2mm]
\node [] (1)   [label=below:$m_1$] {$\circ$};
\node [] (2)  [right=of  1, label=below:$m_2$] {$\circ$};
\node [] (3) [right=of 2, label = below: $m_3$] {$\circ$};
\draw [->]  (3) -- (2);
\end{tikzpicture}~ \nonumber\\[2mm]
1~\text{root} &&
\begin{tikzpicture}[baseline=-2mm]
\node [] (1)   [label=below:$m_1$] {$\circ$};
\node [] (2)  [right=of  1, label=below:$m_2$] {$\circ$};
\node [] (3) [right=of 2, label = below: $m_3$] {$\circ$};
\draw [->]  (2) -- (1);
\draw [->] (3) -- (2);
\end{tikzpicture}~ &&
\begin{tikzpicture}[baseline=-2mm]
\node [] (1)   [label=below:$m_1$] {$\circ$};
\node [] (2)  [right=of  1, label=below:$m_2$] {$\circ$};
\node [] (3) [right=of 2, label = below: $m_3$] {$\circ$};
\draw [->]  (3) -- (2);
\draw [->] (1) -- (2);
\end{tikzpicture}~ \nonumber\\[2mm]
0~\text{roots} &&
\begin{tikzpicture}[baseline=-2mm]
\node [] (1)   [label=below:$m_1$] {$\circ$};
\node [] (2)  [right=of  1, label=below:$m_2$] {$\circ$};
\node [] (3) [right=of 2, label = below: $m_3$] {$\circ$};
\draw [->]  (2) -- (1);
\draw [->] (3) -- (2);
\draw [->] (1) to [out = 20, in = 160] (3);
\end{tikzpicture}~ &&
\begin{tikzpicture}[baseline=-2mm]
\node [] (1)   [label=below:$m_1$] {$\circ$};
\node [] (2)  [right=of  1, label=below:$m_2$] {$\circ$};
\node [] (3) [right=of 2, label = below: $m_3$] {$\circ$};
\draw [->]  (3) -- (2);
\draw [->] (2) to [out = -160, in = -20] (1);
\draw [->] (1) to [out = 20, in = 160] (2);
\end{tikzpicture}~ \nonumber\\
\end{align}
Most of these are smooth and hence define gc families:  just by including the indicated monomials with generic coefficients, we are guaranteed compactness.  This is not the case for the second column with $1$ or $0$ roots: we need additional monomials.  For instance, consider the example
\begin{align}
\begin{tikzpicture}[baseline=-2mm]
\node [] (1)   [label=below:$m_1$] {$\circ$};
\node [] (2)  [right=of  1, label=below:$m_2$] {$\circ$};
\node [] (3) [right=of 2, label = below: $m_3$] {$\circ$};
\draw [->]  (3) -- (2);
\draw [->] (1) -- (2);
\end{tikzpicture}~
&&
W \supset (X_1^{m_1} + X_3^{m_3}) X_2 + X_2^{m_2}~.
\end{align}
Without including another monomial, this will be non-compact.   To see this, we compute
\begin{align}
\left.dW\right|_{X_2 = 0} = (X_1^{m_1} + X_3^{m_3}) dX_2~,
\end{align}
and this clearly has a zero with $X_{1,3} \neq 0$. However, if we can add  $X_1^{p_1} X_3^{p_3}$, then we find a gc family.  This addition is indicated by drawing a ``link'' between the nodes.  
\begin{defn}
Given a potential and its skeleton, a link between two nodes, indicated by a dashed line, denotes a monomial that only depends on the connected fields.
\end{defn}
Thus, in our refined example, we have
\begin{align}
\begin{tikzpicture}[baseline=-2mm]
\node [] (1)   [label=below:$m_1$] {$\circ$};
\node [] (2)  [right=of  1, label=below:$m_2$] {$\circ$};
\node [] (3) [right=of 2, label = below: $m_3$] {$\circ$};
\draw [->]  (3) -- (2);
\draw [->] (1) -- (2);
\draw[dashed] (3) to [out = 120, in =60] (1);
\end{tikzpicture}~
&&
W \supset (X_1^{m_1} + X_3^{m_3}) X_2 + X_2^{m_2} + X_1^{p_1} X_3^{p_3}~.
\end{align}

We need two more graphical concepts before we present the theorem of~\cite{Kreuzer:1992np} on compact gc families.  We already defined a link between two nodes.  Similarly, we can recursively define links between links.
\begin{defn}  A link $L$ between two links $L_1$ and $L_2$ corresponds to a monomial constructed from fields participating in links $L_1$ and $L_2$.  In this case we join $L_1$ and $L_2$ by another dashed line.  In our graphical notation we find it convenient to add a connecting node (indicated by $\bullet$) for every link involved in links between links.
\end{defn}
We also generalize the concept of links and pointers.
\begin{defn} A \textit{pointing link} $L$ is a link between two nodes or two links that points to a node $j$ corresponds to a monomial that is linear in the variable $X_j$ and otherwise depends only on the fields that are joined by the link $L$.    Note that the variable $X_j$ is not counted as one of the variables of the pointing link.  We will also indicate the pointing links with a $\bullet$.
\end{defn}
To illustrate these definitions consider two ways to complete the skeleton
\begin{align}
\begin{tikzpicture}[baseline=-2mm]
\node [] (1)   [label=below:$m_1$] {$\circ$};
\node [] (2)  [right=of  1, label=below:$m_2$] {$\circ$};
\node [] (3) [right=of 2, label = below: $m_3$] {$\circ$};
\node [] (4) [right= of 3, label = below :$m_4$] {$\circ$};
\draw [->]  (3) -- (2);
\draw [->] (1) -- (2);
\end{tikzpicture}
\end{align}
to a gc family:
\begin{align}
\begin{tikzpicture}[baseline=-2mm]
\node [] (1)   [label=below:$m_1$] {$\circ$};
\node [] (2)  [right=of  1, label=below:$m_2$] {$\circ$};
\node [] (3) [right=of 2, label = below: $m_3$] {$\circ$};
\node [] (4) [right= of 3, label = below :$m_4$] {$\circ$};
\draw [->]  (3) -- (2);
\draw [->] (1) -- (2);
\draw [dashed] (3) to [out = 120, in = 60] (1);
\end{tikzpicture}
& &\text{or}  & &
\begin{tikzpicture}[baseline=-2mm]
\node [] (1)   [label=below:$m_1$] {$\circ$};
\node [] (2)  [right=of  1, label=below:$m_2$] {$\circ$};
\node [] (3) [right=of 2, label = below: $m_3$] {$\circ$};
\node [] (4) [right= of 3, label = below :$m_4$] {$\circ$};
\node [] (L1) [above=2mm of 2] {$\bullet$};
\draw [->]  (3) -- (2);
\draw [->] (1) -- (2);
\draw[dashed] (L1) to [out = 0, in =90] (3);
\draw[dashed] (L1) to [out = 180, in = 90] (1);
\draw[dashed] (L1) to [out = 30, in = 150] (4);
\end{tikzpicture}
\end{align}
On the left we have a link between $3$ and $4$, which corresponds to the monomial $X_1^{p_1} X_3^{p_3}$; on the right we have a pointing link, which corresponds to the monomial $X_1^{p_1} X_3^{p_3} X_4$.

We now state the main theorem of~\cite{Kreuzer:1992np}.
\begin{thm} \label{thm:KS} Fix a family of potentials.  A necessary and sufficient condition for this to be gc is that it contains a member $W$ that can be associated a graph with the following properties:
\begin{enumerate}
\item each variable is a root or a pointer;
\item for any pair of variables or for a variable and link that point at some node $v$, there is a link $L$ joining the two pointers, and $L$ does not point at $v$ or any of the nodes that are targets of the objects that are joined by $L$.
\end{enumerate}
\end{thm}

This theorem can be used to describe all possible gc families with fixed $n$.  At each $n$ we construct all the skeleton graphs.  Some will already fulfill the conditions of the theorem and so correspond to gc families with $c < 3n$.  For those skeletons that do not fulfill the conditions of the theorem, we must add links and pointer links in all possible ways so as to fulfill the second criterion of the theorem.  This leads to a complete (in general redundant) list of gc families.


Before we leave the general discussion, we prove one more result that will be useful to us.
\begin{thm}  \label{thm:chargebound} Let $m_i$ denote the exponents for the roots and pointers of a LG skeleton.  Then the total R-charge
\begin{align*}
\qtot = \sum_{i=1}^n q_i   \le \sum_{i=1}^n \frac{1}{m_i}~,
\end{align*}
with equality if and only if every variable is a root.
\end{thm}
\begin{proof}  Let $X_\alpha$, $\alpha = 1,\ldots, r$ denote the roots of the skeleton, and let $X_I$, $I=1,\ldots, n-r$ denote the pointers, with $X_I$ pointing at $X_{\pi(I)}$.  In particular, it means that $W \supset X_\alpha^{m_\alpha}$ for every $\alpha$ and $W \supset X_I^{m_I} X_{\pi(I)}$.  Thus,  $q_\alpha = \frac{1}{m_\alpha}$ and $q_I = \frac{1}{m_I} (1-q_{\pi(I)})$, and
\begin{align*}
\qtot = \sum_{\alpha=1}^r \frac{1}{m_\alpha} + \sum_{I=1}^{n-r} \frac{1}{m_I} (1- q_{\pi(I)}) = \sum_{i=1}^n \frac{1}{m_i} - \sum_{I=1}^{n-r} \frac{q_{\pi(I)}}{m_I}~.
\end{align*}
Since $0<q_{\pi(I)} < 1/2$ for every pointer variable, the result follows.
\end{proof}
This has the following corollary that will be useful in characterizing infinite families of exponents. 
\begin{enumerate}
\item If $n=3$, and for some pair $i\neq j$  $\M{i} + \M{j} <1/2$, then the remaining exponent is bounded.  This only fails if the pair belongs to
\begin{align}
\label{cor:n3}
(m_i,m_j) \in \{ (2,\ast)~,~(3,3)~,~(3,4)~,~(3,5)~,~(3,6)~,~(4,4)\}~.
\end{align}
\item If $n=4$, and for some triplet $i\neq j\neq k$  $\M{i} + \M{j} + \M{k} <1$, then the remaining exponent is bounded.   The only way this fails is if the triplet has the form
\begin{align}
\label{cor:n4}
(m_i,m_j,m_k) \in \{ (k_i,k_j,2)~, (3,3,3)\}~,
\end{align}
where $\frac{1}{k_i} + \frac{1}{k_j} \ge \frac{1}{2}$.
\item If $n=5$, and for some quadruplet $i\neq j \neq k\neq l$  $\M{i} + \M{j} + \M{k} +\M{l} < 3/2$, then the remaining exponent is bounded.  The only way this fails is if the quadruplet has the form
\begin{align}
\label{cor:n5}
(m_i,m_j,m_k,m_l) \in \{ (k_i,k_j,k_k,2)\}~,
\end{align}
where $\frac{1}{k_i} + \frac{1}{k_j} + \frac{1}{k_l} \ge 1$.
\end{enumerate}
We also note that $\qtot$ is polynomial in $\mu_i = \M{i}$, and in fact is linear in each of the $\mu_i$.  It is then easy to see that $\qtot$ is a monotonically decreasing function with respect to each of the exponents.

\section{The feasibility of links} \label{s:links}
Consider a connected graph $G$ on $n$ nodes for a superpotential $W$. The skeleton of the graph determines the charges $q_i$ in terms of the exponents.  If a graph contains a link $L$ linking, say, variables $X_1$ and $X_2$, then $W \supset X_1^{p_1} X_2^{p_2}$, i.e. there exist non-negative integers $p_1,p_2$ such that
\begin{align}
p_1 q_1 + p_2 q_2 = 1~.
\end{align}
This is a non-trivial condition on the exponents, and for many exponents there will be no solution.  We can cast this into a standard counting problem as follows.  Let $q_1 = r_1/d$ and $q_2 = r_2/d$, where $r_1, r_2, d$ are integers and $\gcd(r_1,r_2,d) = 1$.  We then have to solve the equation
\begin{align}
p_1 r_1 + p_2 r_2 = d~
\end{align}
for non-negative $p_1$ and $p_2$.  A necessary condition for a solution is that $\gcd(r_1,r_2) | d$.  When this holds, we can choose $r_1,r_2,d$ so that $\gcd(r_1,r_2) =1$, and we have the classic Frobenius  $2$-coin change problem and its solution:  given two relatively prime positive integers $r_1$ and $r_2$, the largest $d$ which cannot be represented as $p_1 r_1+p_2r_2$ is
\begin{align}
d_{\text{max}} = r_1 r_2 - r_1 -r_2~.
\end{align}
If $d < d_{\text{max}}$ it may or may not be representable; indeed, we have the result that exactly half of the integers between $1$ and $(r_1-1)(r_2-1)$ are representable.\footnote{A really nice description of this classic statement may be found in the first chapter of~\cite{Beck:2015rs}.}

If $G$ contains a pointing link with monomial $W \supset X_1^{p_1} X_2^{p_2} X_3$, there is a similar problem to solve, with
\begin{align}
p_1 q_1 + p_2 q_2 = 1-q_3~.
\end{align}
Hence, we have a collection of Frobenius coin change problems, one for every link.  Each of these is a relatively simple $2$-coin change problem.  

 $G$ may contain links between links (this is the case when two pointing links point to the same node).  When this holds, we must consider a $4$-coin change problem.  In this case there is no known closed form expression for the largest non-representable integer.
We make the following definition.
\begin{defn}
A graph $G$ and a fixed set of exponents $\bM = (m_1,m_2,\ldots,m_n)$ is \textit{feasible} if and only if there is a simultaneous solution to the set of Frobenius coin change problems.
\end{defn}

\subsection{Examples of feasibility constraints} \label{ss:pointedfeasibility}

There are examples of skeletons and exponents for which there are no feasible graphs.  Consider  $n=4$ and charges
\begin{align}
\label{eq:nofeasibility}
q_1 & = \frac{1}{3}~,&
q_2 & = \frac{2}{5}~,&
q_3 & = \frac{4}{25}~,&
q_4 & = \frac{1}{5}~~,& \implies
c & = \frac{136}{25} < 6~.
\end{align}
The most general superpotential is given by
\begin{align}
W & = a_1 X_1^3 + (a_2 X_2^2 +a_3 X_3^5) X_4 + a_4 X_2 X_4^3 + a_5 X_4^5~.
\end{align}
For any compact $W$ the skeleton is constrained as follows: $X_1$ must be a root, $X_2$ and $X_3$ must be pointing at $X_4$.  Finally, $X_4$ is either a root or a pointer at $2$.  In either case there must be a link between $X_2$ and $X_3$; the link may point at $X_1$.   So, we have four possible graphs:
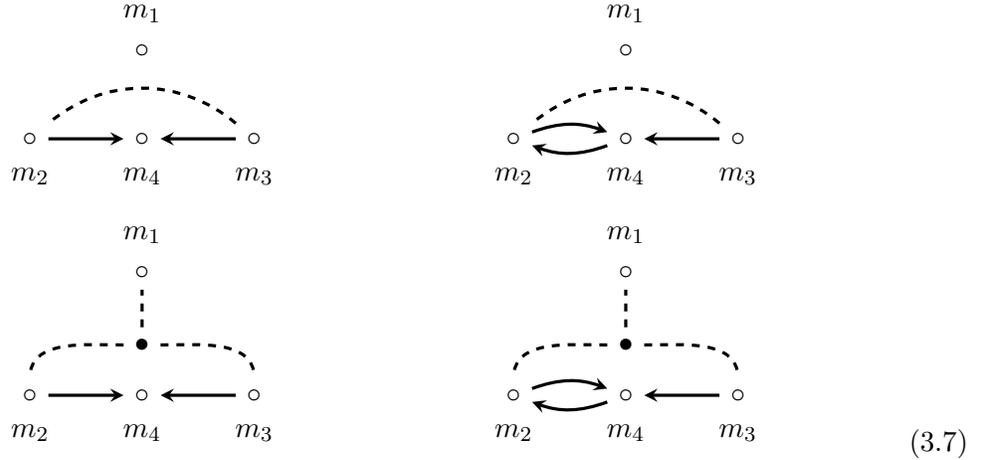
\begin{align}
\label{eq:feasible}
\begin{tikzpicture}
\node [] (1)   [label=below:$m_2$] {$\circ$};
\node [] (2)  [right=of  1, label=below:$m_4$] {$\circ$};
\node [] (3) [right=of 2, label = below: $m_3$] {$\circ$};
\node [] (4) [above= 7mm of 2, label = above :$m_1$] {$\circ$};
\draw [->]  (3) -- (2);
\draw [->] (1) -- (2);
\draw[dashed] (3) to [out = 140, in = 40] (1);
\end{tikzpicture}
&&
\begin{tikzpicture}
\node [] (1)   [label=below:$m_2$] {$\circ$};
\node [] (2)  [right=of  1, label=below:$m_4$] {$\circ$};
\node [] (3) [right=of 2, label = below: $m_3$] {$\circ$};
\node [] (4) [above= 7mm of 2, label = above :$m_1$] {$\circ$};
\draw [->]  (3) -- (2);
\draw [->] (1) to [out=20, in =160] (2);
\draw[->] (2)  to [out=200, in=-20] (1);
\draw[dashed] (3) to [out = 140, in = 40] (1);
%
\end{tikzpicture}
\nonumber\\
\begin{tikzpicture}
\node [] (1)   [label=below:$m_2$] {$\circ$};
\node [] (2)  [right=of  1, label=below:$m_4$] {$\circ$};
\node [] (3) [right=of 2, label = below: $m_3$] {$\circ$};
\node [] (L1) [above=2mm of 2] {$\bullet$};
\node [] (4) [above= 5mm of L1, label = above :$m_1$] {$\circ$};
\draw [->]  (3) -- (2);
\draw [->] (1) -- (2);
\draw[dashed] (L1) to [out = 0, in =90] (3);
\draw[dashed] (L1) to [out = 180, in = 90] (1);
\draw[dashed] (L1) to [out = 90, in = -90] (4);
\end{tikzpicture}
&&
\begin{tikzpicture}
\node [] (1)   [label=below:$m_2$] {$\circ$};
\node [] (2)  [right=of  1, label=below:$m_4$] {$\circ$};
\node [] (3) [right=of 2, label = below: $m_3$] {$\circ$};
\node [] (L1) [above=2mm of 2] {$\bullet$};
\node [] (4) [above= 5mm of L1, label = above :$m_1$] {$\circ$};
\draw [->]  (3) -- (2);
\draw [->] (1) to [out=20, in =160] (2);
\draw[->] (2)  to [out=200, in=-20] (1);
\draw[dashed] (L1) to [out = 0, in =90] (3);
\draw[dashed] (L1) to [out = 180, in = 90] (1);
\draw[dashed] (L1) to [out = 90, in = -90] (4);
\end{tikzpicture}
\end{align}
However, with the charges as in~(\ref{eq:nofeasibility}) no graph is feasible.

As we will see below, a graph $G$ often allows for an infinite set of exponents consistent with $c<6$; however, it may be that only a finite subset of those exponents is feasible.  For instance, we give an example of an infinite family of exponents where only one member is feasible.  We start with the skeleton superpotential
\begin{align}
W = (X_1^4 + X_2^4) X_3 + X_3^k \implies q_1 = q_2 = \frac{k-1}{4k}~~~,q_3 = \frac{1}{k}~,
\end{align}
with $k\ge 4$.
To produce a gc family from this skeleton we need a link between $X_1$ and $X_2$, which is possible if and only if
\begin{align}
s = \frac{4k}{k-1}~
\end{align}
is a positive integer.  It is easy to see that $k=5$ is the only feasible value, in which case this LG family is deformation--equivalent to the simpler and decomposable $W = X_1^5 + X_2^5 + X_3^5$.

As another example of linking subtleties, we can see that for a fixed skeleton and choice of exponents it may be that feasible graphs necessarily involve a pointing link.  Consider the charges
\begin{align}
q_1 & = \frac{1}{3}~,&
q_2 & = \frac{2}{5}~,&
q_3 & = \frac{4}{45}~,&
q_4 & = \frac{1}{5}~,&\implies c  = \frac{88}{15} < 6~.
\end{align}
In this case, the most general superpotential is
\begin{align}
W = a_1X_1^3 + (  a_2 X_2^2+a_3X_3^9) X_4 + (a_4 X_2 X_3^3 + a_5 X_3^3 X_4^2) X_1 + a_6 X_4^3 X_2 + a_7 X_4^5~.
\end{align}
There are feasible graphs (the two bottom graphs in~(\ref{eq:feasible})), but they necessarily involve a pointing link.  Note that the list of singularities in~\cite{Yau:2003sst} does not involve any pointing links (no family contains a monomial in more than two variables) and is therefore incomplete. 

\subsection{Pointing links and disconnected skeletons}
The combinatoric challenge of LG classification is in the analysis of links and pointing links.  There is one aspect of the problem that is simpler than it might first appear.  While pointing links can connect distinct connected components of a skeleton graph, for $n\le 5$ the possibilities for this are limited.

As we will see in the next section from the classification of connected skeletons, for $n\le 2$ there are no links, while for $n\le 3$ there is at most one link.  An $n=4$ connected skeleton has at most $3$ links.  Since we must have $n\le 5$ in total, there are just a few possibilities in which a pointing link can connect  disconnected components of a skeleton:
\begin{enumerate}
\item a connected skeleton with $n=3$ has a link that points at the $n=1$ skeleton;
\item a connected skeleton with $n=3$ has a link that points at a connected $n=2$ skeleton;
\item a connected skeleton with $n=4$ has $k=1,2,3$ links that point at the $n=1$ skeleton.  
\end{enumerate}

\section{Connected skeletons in $n \le 5$ variables} \label{s:skeletons}
In this section we classify all connected LG skeletons in $n\le 5$ variables.  These results are surely known to many; see, for instance~\cite{Hertling:2012ku}, however, we have not been able to find the $n=5$ case in the literature.

A LG skeleton is a directed graph with $n$ nodes and no edge from a node to itself. For each node let $\ell_{\text{in}}$ and $\ell_{\text{out}}$ denote, respectively, the number of incoming and outgoing edges.  In a LG skeleton each node has at most one outgoing edge: $\ell_{\text{out}}(v) = 0$ when $v$ is a root, and $\ell_{\text{out}}(v) = 1$ when $v$ is a pointer.  If the number of roots $r \ge 2$ then the skeleton graph is necessarily disconnected.  For a LG skeleton with $p$ pointers we split the nodes $v$ into the $r=n-p$ roots $\{u_1,\ldots,u_r\}$ and the $p$ pointers $\{w_1,\ldots,w_r\}$.  Evidently
\begin{align}
\sum_{v} \ell_{\text{in}}(v)  = \sum_{u} \ell_{\text{in}}(u) + \sum_{w} \ell_{\text{in}} (w) = p~.
\end{align}
Thus, every LG skeleton with $p$ pointers yields a partition of $p$ into $n$ non-negative integers 
$[\ell_{\text{in}}(u_1),\ldots,\ell_{\text{in}}(u_r);\ell_{\text{in}}(w_1),\ldots,\ell_{\text{in}}(w_p)]$.  Without loss of generality we may assume $\ell_{\text{in}}(u_1) \ge \ell_{\text{in}}(u_2)\ge\cdots \ge \ell_{\text{in}}(u_r)$ and similarly for the pointers.  


Suppose $G$ is a graph for a gc LG family.  When $G$ is disconnected, there is a locus in the parameter space where the theory is compact and can be written as a direct product of SCFTs, one for each connected component.  Hence, to classify the LG theories up to marginal deformations we can restrict attention to connected graphs.

Every graph $G$ has a skeleton subgraph $S$ obtained by dropping all of the links in $G$.   The resulting skeleton may either be connected or disconnected.  Turning this around, given a disconnected skeleton with components $S_1$ and $S_2$, the only element in the construction that might yield a connected graph $G$ involves a pointing link from either $S_1$ to $S_2$ or from $S_2$ to $S_1$.  So, to classify the indecomposable LG SCFTs with $c <6$ we can proceed as follows:
\begin{enumerate}
\item  construct all connected skeletons with $n \le 5$ nodes and $1$ root;
\item construct all connected skeletons with $n\le 5$ nodes and $0$ roots;
\item for every singular skeleton (i.e. those with $\ell_{\text{in}}(v) >1$) add the required links;
\item find restrictions on exponents for $n=3,4,5$ such that $c<6$ and every link is feasible;
\item for every singular skeleton allow for any possible pointing links from the skeleton to itself (this will in general require links between links) and check for feasibility;
\item construct all connected graphs where a disconnected skeleton with $c<6$ is made into a connected graph by pointing links.
\end{enumerate}
In what follows we will carry out these steps with one important simplification: we will only be interested in indecomposable theories.  An indecomposable gc--family has charges that do not allow $W$ to be written as a sum of two non-interacting terms for any values of the complex coefficients away from the singular locus.  We will work with a slightly more general notion.
\begin{defn}  We will say that a graph is reducible if one of the following holds:
\begin{enumerate}
\item the graph has more than $1$ root;
\item the graph has $1$ root, but the exponents are such that some pointer variable can appear in the superpotential as a root ;
\item the graph has $0$ roots, but the exponents are such that some pointer variable can appear in the superpotential as a root.
\end{enumerate}
\end{defn}
In the first and second case the gc--family is decomposable; in the third case if precisely one pointer can appear as a root, the gc--family may not be decomposable, but it is then necessarily equivalent to a gc--family obtained from a graph with $1$ root.  So, to classify indecomposable theories by our procedure it is sufficient to describe irreducible theories.

In the remainder of this section we will construct all the connected skeletons, and we will indicate all the required links, but we will not illustrate the possible pointing links in the diagrams.  In the following sections we will discuss the constraints on exponents, combinatorics of links and pointing links.

\subsection{The connected skeletons  with $n\le 5$}
We will denote the connected skeletons in $k$ nodes by $S_{k,\alpha}$, where $\alpha$ is a label for the graph type at fixed $k$.  For each skeleton we will include the links necessary to yield a gc family, but at this point we will not distinguish the pointing links.

\subsubsection*{$n\le 2$ skeletons}
These are straightforward.  There are no links, and the skeletons are
\begin{align}
\label{eq:n2LG}
\arraycolsep=5mm
\begin{array}{c|c|c}
S_{1,1}	&  S_{2,1}		& S_{2,2} \\[2mm] \hline
\begin{tikzpicture}[baseline=1mm]
\node [] (1) [label=below:$m_1$]{$\circ$};
\end{tikzpicture} &
\begin{tikzpicture}[baseline=1mm]
\node [] (1) [label=below:$m_1$]{$\circ$};
\node [] (2) [right= of 1, label=below:$m_2$]{$\circ$};
\draw[->] (1) -- (2);
\end{tikzpicture} &
\begin{tikzpicture}[baseline=1mm]
\node[] (1) [label=below:$m_1$]{$\circ$};
\node[] (2) [right = of 1, label = below:$m_2$]{$\circ$};
\draw[->] (1) -- (2);
\draw[->] (2) to [out =150, in =30] (1);
\end{tikzpicture} \\
W = X_1^m &
W = X_1^{m_1}X_2 + X_2^{m_2} &
W = X_1^{m_1}X_2 + X_2^{m_2}X_1 \\
\begin{aligned}q_1 = \frac{1}{m_1} 	\end{aligned}
&\begin{aligned} q_1& = \frac{m_2 -1}{m_1m_2} \\ q_2& = \frac{1}{m_2} \end{aligned} 
&\begin{aligned} q_1 &= \frac{m_2 -1}{m_1m_2-1} \\	q_2 &= \frac{m_1-1}{m_1m_2-1}  \end{aligned}
\end{array}
\end{align}

\subsubsection*{$n=3$ skeletons}
These are easy enough to construct directly, but with the $n=4,5$ cases in view, let us describe a procedure for obtaining all of the connected skeletons.  We first note that a connected skeleton has either $1$ root or $0$ roots.  Second, each node $v$ has $\ell_{\text{in}}(v)$ arrows pointing to it, which requires $\binom{\ell_{\text{in}}(v)}{2}$ links.  So, we can organize the skeletons by the number of roots and the partition of $p$ among the $\ell_{\text{in}}(v)$.
Applying this to $n=3$, we first consider the case of $1$ root, which can have at most $1$ link:
\begin{align}
\arraycolsep=5mm
\begin{array}{cc}
S_{3,1} & S_{3,2} \\[2mm] \hline
\begin{tikzpicture}
\node[] (1) [label=below:$m_1$]{$\circ$};
\node[] (2) [right =of 1, label = below:$m_2$]{$\circ$};
\node[] (3) [right = of 2, label = below:$m_3$]{$\circ$};
\draw[->] (1)--(2);
\draw[->] (2)--(3);
\end{tikzpicture} &
\begin{tikzpicture}
\node[] (1) [label=below:$m_1$]{$\circ$};
\node[] (2) [right =of 1, label = below:$m_2$]{$\circ$};
\node[] (3) [right = of 2, label = below:$m_3$]{$\circ$};
\draw[->] (1)--(2);
\draw[->] (3)--(2);
\draw[dashed] (1) to [out = 45, in = 135] (3);
\end{tikzpicture}\\[2mm]
X_1^{m_1}X_2 + X_2^{m_2}X_3 + X_3^{m_3} 
&
X_1^{m_1}X_2 + X_2^{m_2} + X_3^{m_3}X_2+ X_1^{p_1}X_3^{p_3}
\\[2mm]
\begin{aligned}
q_1 &= \frac{m_2m_3-m_3+1}{m_1m_2m_3} \\[2mm]
q_2 &= \frac{m_3-1}{m_2m_3}\\[2mm]
q_3 &= \frac{1}{m_3}  
\end{aligned}
&
\begin{aligned}
q_1 &= \frac{m_2 -1}{m_1m_2} \\[2mm]
q_2 &= \frac{1}{m_2}\\[2mm]
q_3 &= \frac{m_2-1}{m_2m_3}
\end{aligned}
\end{array}
\end{align}
Next, we consider the cases without roots.
\begin{align}
\arraycolsep=5mm
\begin{array}{cc}
S_{3,3} & S_{3,4}\\[2mm]
\begin{tikzpicture}
\node[] (1) [label=above:$m_1$]{$\circ$};
\node[] (2) [below left = 3mm and 3mm of 1,label=below:$m_2$]{$\circ$};
\node[] (3) [below right = 3mm and 3mm  of 1, label=below:$m_3$]{$\circ$};
\draw[->] (1)--(2);
\draw[->] (2)--(3);
\draw[->] (3) -- (1);
\end{tikzpicture} &
\begin{tikzpicture}
\node[] (1)  [label=below:$m_1$]{$\circ$};
\node[] (2) [right of =1, label=below:$m_2$]{$\circ$};
\node[] (3) [right of=2,label=below:$m_3$]{$\circ$};
\draw[->] (1)--(2) ;
\draw[->] (2) to [out = -135, in = -45] (1);
\draw[->] (3) -- (2);
\draw[dashed] (1) to [out = 45, in = 135] (3);
\end{tikzpicture} \\[2mm]
X_1^{m_1}X_2 + X_2^{m_2}X_3 + X_3^{m_3}X_1 &
X_1^{m_1}X_2 + X_2^{m_2}X_1 + X_3^{m_3}X_2 + X_1^{p_1}X_3^{p_3}\\[2mm]
\begin{aligned}
q_1 &= \frac{1-m_3 + m_2m_3}{m_1m_2m_3+1} \\[2mm]
q_2 &= \frac{1 - m_1 + m_1m_3}{m_1m_2m_3 + 1} \\[2mm]
q_3 &= \frac{1-m_2 + m_1m_2}{m_1m_2m_3 + 1}
\end{aligned}&
\begin{aligned}
q_1 &= \frac{m_2-1}{m_1m_2-1}\\[2mm]
q_2 &= \frac{m_1 -1}{m_1m_2 -1}\\[2mm]
q_3 &= \frac{m_1m_2 - m_1}{m_3(m_1m_2-1)}
\end{aligned}
\end{array}
\end{align}

\subsubsection*{$n=4$ skeletons}
At this point we hope that the reader sees the translation from the graphic notation to the superpotential.  It is then easy to find the charges in terms of the exponents.  So, for the remaining cases we will simply present the graphs.  

The skeletons with $1$ root, organized by the number of links, are
\begin{align}
\arraycolsep=5mm
\begin{array}{cc}
S_{4,1}	&	S_{4,2}  \\[2mm]
\begin{tikzpicture}
\node[] (1)  [label=below:$m_4$]{$\circ$};
\node[] (2) [right of =1, label=below:$m_1$]{$\circ$};
\node[] (3) [right of=2,label=below:$m_2$]{$\circ$};
\node[] (4) [right of=3,label=below:$m_3$]{$\circ$};
\draw[->] (1)--(2) ;
\draw[->] (2) --(3);
\draw[->] (3) -- (4);
\end{tikzpicture}
&
\begin{tikzpicture}
\node[] (2) [ label=below:$m_1$]{$\circ$};
\node[] (3) [right of=2,label=below:$m_2$]{$\circ$};
\node[] (4) [right of=3,label=below:$m_3$]{$\circ$};
\node[] (1)  [right of = 4, label=below:$m_4$]{$\circ$};
\draw[->] (1)--(4) ;
\draw[->] (2) --(3);
\draw[->] (4) -- (3);
\draw[dashed] (2) to [out = 45, in = 135] (4);
\end{tikzpicture} \\[4mm]
S_{4,3}	&	S_{4,4} \\[2mm]
\begin{tikzpicture}
\node[] (3)  [label=below:$m_2$]{$\circ$};
\node[] (1) [above left = 3mm and 3mm of 3, label=above:$m_4$]{$\circ$};
\node[] (2) [below left = 3mm and 3mm of 3,label=below:$m_1$]{$\circ$};
\node[] (4) [right of=3,label=below:$m_3$]{$\circ$};
\draw[->] (2) -- (3);
\draw[->] (1) --(3);
\draw[->] (3) -- (4);
\draw[dashed] (1) -- (2);
\end{tikzpicture}
&
\begin{tikzpicture}
\node[] (4)  [label=below:$m_2$]{$\circ$};
\node[] (1) [below left = 5mm and 5mm  of 4, label=below:$m_1$]{$\circ$};
\node[] (2) [above =5mm of 4,label=above:$m_4$]{$\circ$};
\node[] (3) [below right =5mm and 5mm of 4,label=below:$m_3$]{$\circ$};
\draw[->] (2) -- (4);
\draw[->] (1) --(4);
\draw[->] (3) -- (4);
\draw[dashed] (1) -- (2);
\draw[dashed] (2) -- (3);
\draw[dashed] (1) -- (3);
\end{tikzpicture}
\end{array}
\end{align}
The skeletons without roots are
\begin{align}
\arraycolsep=4mm
\begin{array}{cc}
S_{4,5}	& S_{4,6}\\[2mm]
\begin{tikzpicture}
\node[] (1)  [label=below:$m_1$]{$\circ$};
\node[] (2) [right = of 1, label=below:$m_2$]{$\circ$};
\node[] (3) [above = of  2,label=above:$m_3$]{$\circ$};
\node[] (4) [left =of 3,label=above:$m_4$]{$\circ$};
\draw[->] (1) -- (2);
\draw[->] (2) --(3);
\draw[->] (3) -- (4);
\draw[->] (4)--(1);
\end{tikzpicture} &
\begin{tikzpicture}
\node[] (1)  [label=below:$m_4$]{$\circ$};
\node[] (2) [right = of 1, label=below:$m_3$]{$\circ$};
\node[] (3) [right = of  2,label=below:$m_2$]{$\circ$};
\node[] (4) [right =of 3,label=below:$m_1$]{$\circ$};
\draw[->] (1) -- (2);
\draw[->] (2) --(3);
\draw[->] (3) -- (4);
\draw[->] (4) to [out = 135, in= 45] (3);
\draw[dashed] (4) to [out = -135, in = -45] (2);
\end{tikzpicture} \\[4mm]
S_{4,7}	&S_{4,8}\\[2mm]
\begin{tikzpicture}
\node[] (4) [label=above:$m_4$]{$\circ$};
\node[] (1) [below = 3mmof 4, label=left:$m_1$]{$\circ$};
\node[] (2) [below left = 3mm and 2mm of 1,label=below:$m_2$]{$\circ$};
\node[] (3) [below right = 3mm and 2mm  of 1, label=below:$m_3$]{$\circ$};
\draw[->] (4)--(1);
\draw[->] (1)--(2);
\draw[->] (2)--(3);
\draw[->] (3) -- (1);
\draw[dashed] (3) -- (4);
\end{tikzpicture} &
\begin{tikzpicture}
\node[] (1) [label=below:$m_3$]{$\circ$};
\node[] (2) [right =of 1, label = below: $m_2$]{$\circ$};
\node[] (3) [right = of 2, label = below: $m_1$]{$\circ$};
\node[] (4) [right = of 3, label = below: $m_4$]{$\circ$};
\draw[->] (1) -- (2);
\draw[->] (4) -- (3);
\draw[->] (2)--(3);
\draw[->] (3) to [out = 145, in = 35] (2);
\draw[dashed] (1) to [out = 45, in = 135] (3);
\draw[dashed] (2) to [out = -45, in = -135] (4);
\end{tikzpicture} 
\end{array}
\end{align}
as well as
\begin{align}
S_{4,9} \qquad\quad~~ \nonumber\\[2mm]
\begin{tikzpicture}
\node[] (1) [label=below: $m_4$]{$\circ$};
\node[] (2) [right = of 1, label = below: $m_3$]{$\circ$};
\node[] (3) [right = of 2, label = below: $m_1$]{$\circ$};
\node[] (4) [above = of 2, label = above: $m_2$]{$\circ$};
\draw[->] (1) -- (4);
\draw[->] (2) -- (4);
\draw[->] (3) -- (4);
\draw[->] (4) to [out = 0, in = 90] (3);
\draw[dashed] (1) -- (2);
\draw[dashed] (2) -- (3);
\draw[dashed] (1) to [out = -45, in = -135] (3);
\end{tikzpicture}
\end{align}

\subsubsection*{$n=5$ skeletons}
The skeletons with $1$ root and $1$ link are
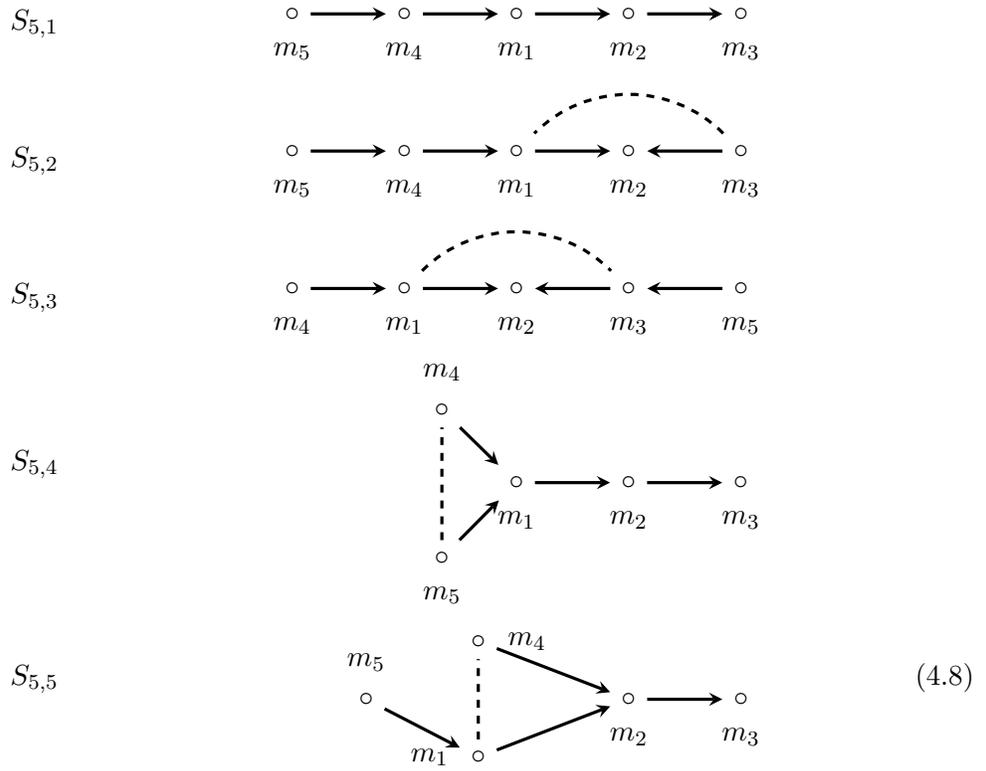
\begin{align}
S_{5,1} && \begin{tikzpicture}[baseline=-2mm]
\node[] (1) [label=below: $m_5$]{$\circ$};
\node[] (2) [right = of 1, label = below: $m_4$]{$\circ$};
\node[] (3) [right = of 2, label = below: $m_1$]{$\circ$};
\node[] (4) [right = of 3, label = below: $m_2$]{$\circ$};
\node[] (5) [right = of 4, label = below: $m_3$]{$\circ$};
\draw[->] (1) -- (2);
\draw[->] (2) -- (3);
\draw[->] (3) -- (4);
\draw[->] (4) -- (5);
\end{tikzpicture}  \nonumber\\
S_{5,2} && 
\begin{tikzpicture}[baseline=-2mm]
\node[] (1) [label=below: $m_5$]{$\circ$};
\node[] (2) [right = of 1, label = below: $m_4$]{$\circ$};
\node[] (3) [right = of 2, label = below: $m_1$]{$\circ$};
\node[] (4) [right = of 3, label = below: $m_2$]{$\circ$};
\node[] (5) [right = of 4, label = below: $m_3$]{$\circ$};
\draw[->] (1) -- (2);
\draw[->] (2) -- (3);
\draw[->] (3) -- (4);
\draw[->] (5) -- (4);
\draw[dashed] (5) to [out = 135, in = 45] (3);
\end{tikzpicture} \nonumber\\
S_{5,3} &&
\begin{tikzpicture}[baseline=-2mm]
\node[] (1) [label=below: $m_4$]{$\circ$};
\node[] (2) [right = of 1, label = below: $m_1$]{$\circ$};
\node[] (3) [right = of 2, label = below: $m_2$]{$\circ$};
\node[] (4) [right = of 3, label = below: $m_3$]{$\circ$};
\node[] (5) [right = of 4, label = below: $m_5$]{$\circ$};
\draw[->] (1) -- (2);
\draw[->] (2) -- (3);
\draw[->] (4) -- (3);
\draw[->] (5) -- (4);
\draw[dashed] (2) to [out = 45, in = 135] (4);
\end{tikzpicture} \nonumber\\
S_{5,4} &&
\begin{tikzpicture}[baseline=-8mm]
\node[] (1) [label=above: $m_4$]{$\circ$};
\node[] (2) [below =15mm of 1, label =below: $m_5$]{$\circ$};
\node[] (3) [below right = 5mm and 5mm of 1, label = below: $m_1$]{$\circ$};
\node[] (4) [right = 10mm of 3, label = below: $m_2$]{$\circ$};
\node[] (5) [right = 10mm of 4, label = below: $m_3$]{$\circ$};
\draw[->] (1) -- (3);
\draw[->] (2) -- (3);
\draw[->] (3) -- (4);
\draw[->] (4) -- (5);
\draw[dashed] (2) -- (1);
\end{tikzpicture} \nonumber\\
S_{5,5}  &&
\begin{tikzpicture}[baseline=2mm]
\node[] (1) [label=above: $m_5$]{$\circ$};
\node[] (2) [below right = 3mm and 10mm of 1, label =left: $m_1$]{$\circ$};
\node[] (3) [above right = 3mm and 10mm of 1, label = right: $m_4$]{$\circ$};
\node[] (4) [right = 30mm of 1, label = below: $m_2$]{$\circ$};
\node[] (5) [right =10mm of 4, label = below: $m_3$]{$\circ$};
\draw[->] (1) -- (2);
\draw[->] (2) -- (4);
\draw[->] (3) -- (4);
\draw[->] (4) -- (5);
\draw[dashed] (2) -- (3);
\end{tikzpicture}
\end{align} 
There is one skeleton with two links and one with three links:
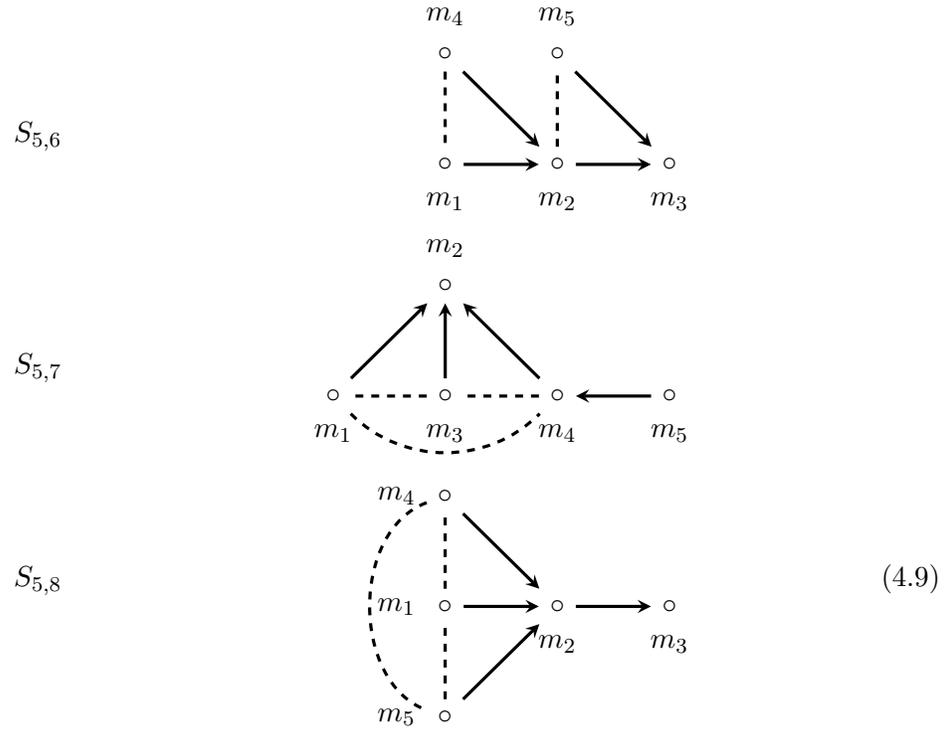
\begin{align}
S_{5,6}&&
\begin{tikzpicture}[baseline=3mm]
\node[] (1) [label=below: $m_1$]{$\circ$};
\node[] (2) [above = of 1, label =above: $m_4$]{$\circ$};
\node[] (3) [right = of 1, label = below: $m_2$]{$\circ$};
\node[] (4) [above= of 3, label = above: $m_5$]{$\circ$};
\node[] (5) [right = of 3, label = below: $m_3$]{$\circ$};
\draw[->] (1) -- (3);
\draw[->] (2) -- (3);
\draw[->] (3) -- (5);
\draw[->] (4) -- (5);
\draw[dashed] (2) -- (1);
\draw[dashed] (3) -- (4);
\end{tikzpicture} \nonumber\\
S_{5,7} &&
\begin{tikzpicture}[baseline=3mm]
\node[] (1) [label=below: $m_1$]{$\circ$};
\node[] (2) [right = of 1, label =below: $m_3$]{$\circ$};
\node[] (3) [right = of 2, label = below: $m_4$]{$\circ$};
\node[] (4) [right = of 3, label = below: $m_5$]{$\circ$};
\node[] (5) [above =of 2, label = above: $m_2$]{$\circ$};
\draw[->] (1) -- (5);
\draw[->] (2) -- (5);
\draw[->] (3) -- (5);
\draw[->] (4) -- (3);
\draw[dashed] (2) -- (1);
\draw[dashed] (3) -- (2);
\draw[dashed] (1) to [out = -45, in = -135] (3);
\end{tikzpicture} \nonumber\\
S_{5,8} &&
\begin{tikzpicture}[baseline=-12mm]
\node[] (1) [label=left: $m_4$]{$\circ$};
\node[] (2) [below = of 1, label = left: $m_1$]{$\circ$};
\node[] (3) [below = of 2, label = left: $m_5$]{$\circ$};
\node[] (4) [right = of 2, label = below: $m_2$]{$\circ$};
\node[] (5) [right =of 4, label = below: $m_3$]{$\circ$};
\draw[->] (1) -- (4);
\draw[->] (2) -- (4);
\draw[->] (3) -- (4);
\draw[->] (4) -- (5);
\draw[dashed] (2) -- (1);
\draw[dashed] (3) -- (2);
\draw[dashed] (1) to [out = -160, in = 160] (3);
\end{tikzpicture}
\end{align}
Finally, we have a skeleton with six links:
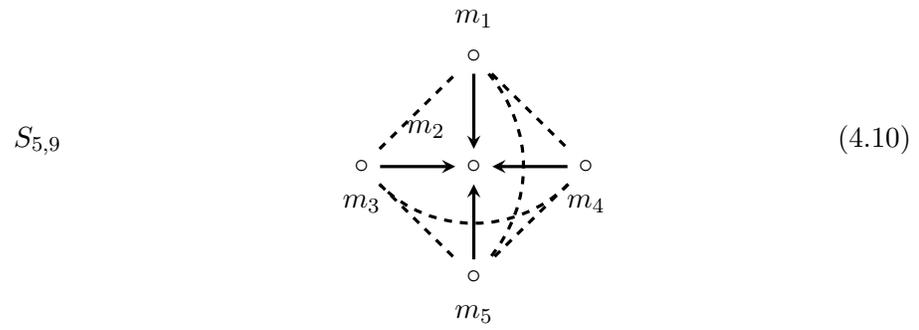
\begin{align}
S_{5,9} &&
\begin{tikzpicture}[baseline=-12mm]
\node[] (1) [label=above: $m_1$]{$\circ$};
\node[] (5) [below =of 1, label = above left: $m_2$]{$\circ$};
\node[] (2) [left = of 5, label =below: $m_3$]{$\circ$};
\node[] (3) [right = of 5, label = below: $m_4$]{$\circ$};
\node[] (4) [below = of 5, label = below: $m_5$]{$\circ$};
\node[] (6) [right = 2mm of 3] {};
\draw[->] (1) -- (5);
\draw[->] (2) -- (5);
\draw[->] (3) -- (5);
\draw[->] (4) -- (5);
\draw[dashed] (1) -- (3);
\draw[dashed] (3) to [out = -135, in = -45] (2);
\draw[dashed] (1) to [out = -50, in = 50] (4);
\draw[dashed] (3) -- (4);
\draw[dashed] (2) -- (4);
\draw[dashed] (2) -- (1);
\end{tikzpicture}
\end{align}
We now move on to skeletons with $0$ roots.  Those with one link or less are
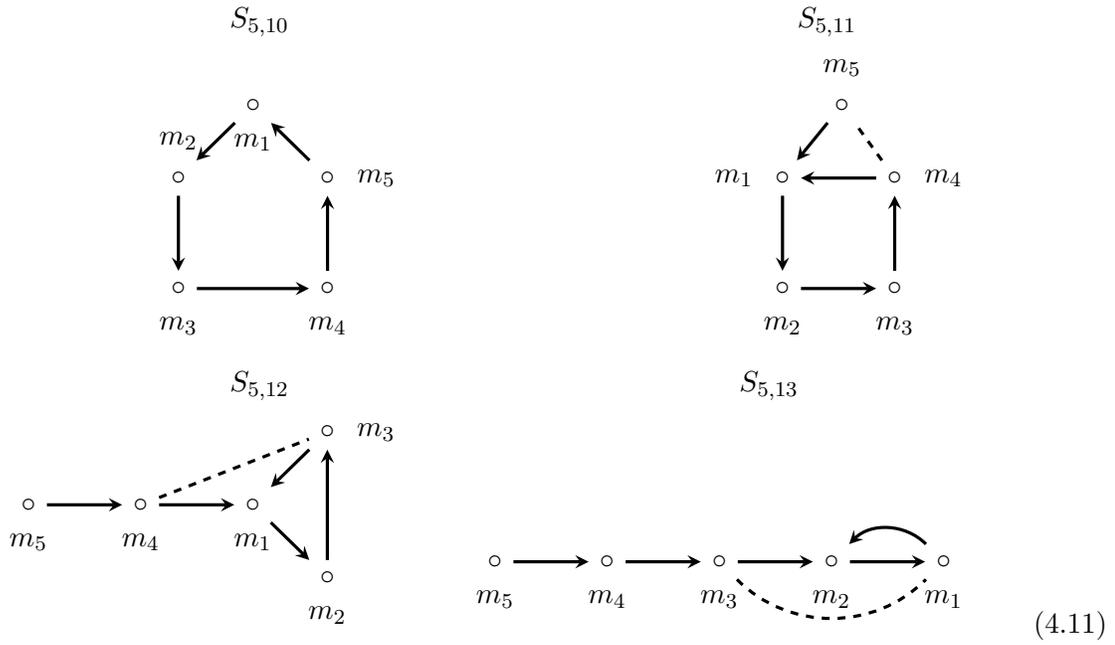
\begin{align}
S_{5,10} \qquad \qquad&&  S_{5,11}\qquad\qquad  \nonumber\\
\begin{tikzpicture}
\node[] (1) [label=below: $m_1$]{$\circ$};
\node[] (2) [below left = 5mm and 5mm of 1, label =above: $m_2$]{$\circ$};
\node[] (3) [below = of 2, label = below: $m_3$]{$\circ$};
\node[] (5) [below right  = 5mm and 5mm of 1, label = right: $m_5$]{$\circ$};
\node[] (4) [below= of 5, label = below: $m_4$]{$\circ$};
\draw[->] (1) -- (2);
\draw[->] (2) -- (3);
\draw[->] (3) -- (4);
\draw[->] (4) -- (5);
\draw[->] (5) -- (1);
\end{tikzpicture} 
 &&
\begin{tikzpicture}
\node[] (1) [label=left: $m_1$]{$\circ$};
\node[] (2) [below = of 1, label =below: $m_2$]{$\circ$};
\node[] (3) [right = of 2, label = below: $m_3$]{$\circ$};
\node[] (4) [above = of 3, label = right: $m_4$]{$\circ$};
\node[] (5) [above right =5mm and 3mm  of 1, label = above: $m_5$]{$\circ$};
\draw[->] (1) -- (2);
\draw[->] (2) -- (3);
\draw[->] (3) -- (4);
\draw[->] (4) -- (1);
\draw[->] (5) -- (1);
\draw[dashed] (4) -- (5);
\end{tikzpicture} \nonumber\\
S_{5,12}\qquad\qquad &&  S_{5,13} \qquad\qquad\qquad\nonumber\\
\begin{tikzpicture}
\node[] (1) [label=below: $m_5$]{$\circ$};
\node[] (2) [right = 10mm of 1, label =below: $m_4$]{$\circ$};
\node[] (3) [right = 10mm of 2, label = below: $m_1$]{$\circ$};
\node[] (4) [below right = 5mm and 5mm  of 3, label = below: $m_2$]{$\circ$};
\node[] (5) [above right =5mm and 5mm  of 3, label = right: $m_3$]{$\circ$};
\draw[->] (1) -- (2);
\draw[->] (2) -- (3);
\draw[->] (3) -- (4);
\draw[->] (4) -- (5);
\draw[->] (5) -- (3);
\draw[dashed] (2) -- (5);
\end{tikzpicture} &&
\begin{tikzpicture}
\node[] (1) [label=below: $m_5$]{$\circ$};
\node[] (2) [right = of 1, label =below: $m_4$]{$\circ$};
\node[] (3) [right = of 2, label = below: $m_3$]{$\circ$};
\node[] (4) [right =  of 3, label = below: $m_2$]{$\circ$};
\node[] (5) [right = of 4, label = below: $m_1$]{$\circ$};
\draw[->] (1) -- (2);
\draw[->] (2) -- (3);
\draw[->] (3) -- (4);
\draw[->] (4) -- (5);
\draw[->] (5) to [out = 135, in = 45] (4);
\draw[dashed] (3) to [out = -45, in = -135] (5);
\end{tikzpicture} 
\end{align}
Next, the skeletons with two links are
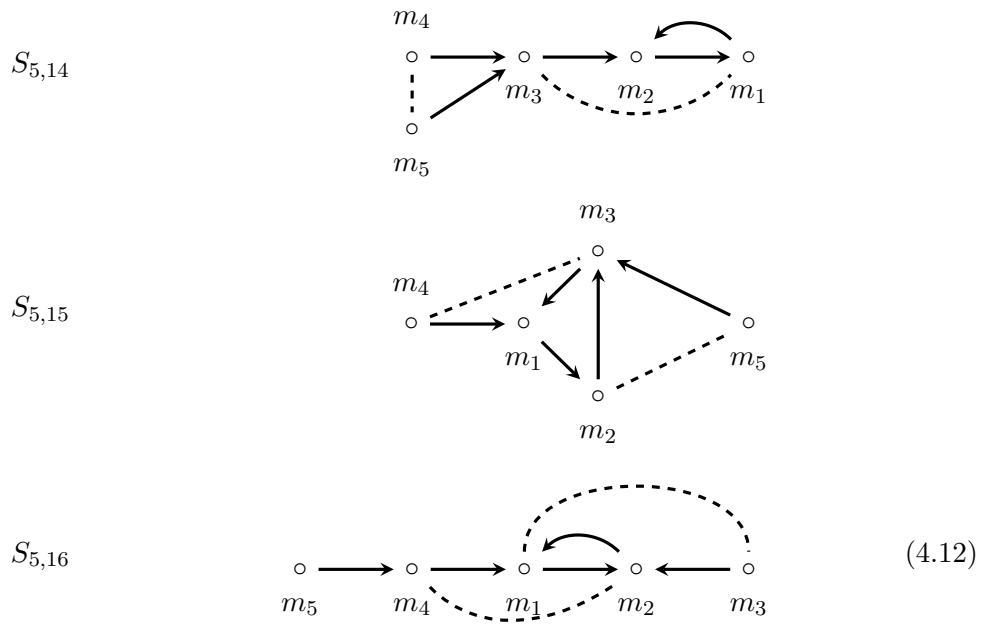
\begin{align}
S_{5,14} &&
\begin{tikzpicture}[baseline=-2mm]
\node[] (1) [label=above: $m_4$]{$\circ$};
\node[] (2) [below = 5mm of 1, label =below: $m_5$]{$\circ$};
\node[] (3) [right = of 1, label = below: $m_3$]{$\circ$};
\node[] (4) [right =  of 3, label = below: $m_2$]{$\circ$};
\node[] (5) [right = of 4, label = below: $m_1$]{$\circ$};
\draw[->] (1) -- (3);
\draw[->] (2) -- (3);
\draw[->] (3) -- (4);
\draw[->] (4) -- (5);
\draw[->] (5) to [out = 135, in = 45] (4);
\draw[dashed] (3) to [out = -45, in = -135] (5);
\draw[dashed] (1) -- (2);
\end{tikzpicture} \nonumber\\
S_{5,15} &&
\begin{tikzpicture}[baseline=1mm]
\node[] (1) [label=above: $m_4$]{$\circ$};
\node[] (2) [right = 10mm of 1, label =below: $m_1$]{$\circ$};
\node[] (3) [below right = 5mm and 5mm of 2, label = below: $m_2$]{$\circ$};
\node[] (4) [above right = 5mm and 5mm of 2, label = above: $m_3$]{$\circ$};
\node[] (5) [right = 2.5cm of 2, label = below: $m_5$]{$\circ$};
\draw[->] (1) -- (2);
\draw[->] (2) -- (3);
\draw[->] (3) -- (4);
\draw[->] (4) -- (2);
\draw[->] (5) -- (4);
\draw[dashed] (3) --(5);
\draw[dashed] (1) -- (4);
\end{tikzpicture}\nonumber\\
 S_{5,16} &&
\begin{tikzpicture}[baseline=1mm]
\node[] (1) [label=below: $m_5$]{$\circ$};
\node[] (2) [right= of 1, label =below: $m_4$]{$\circ$};
\node[] (3) [right = of 2, label = below: $m_1$]{$\circ$};
\node[] (4) [right = of 3, label = below: $m_2$]{$\circ$};
\node[] (5) [right = of 4, label = below: $m_3$]{$\circ$};
\draw[->] (1) -- (2);
\draw[->] (2) -- (3);
\draw[->] (3) -- (4);
\draw[->] (4) to [out = 135, in = 45] (3);
\draw[->] (5) -- (4);
\draw[dashed] (2) to [out = -45, in = -145] (4);
\draw[dashed] (3) to [out = 90, in = 90] (5);
\end{tikzpicture}
\end{align}
The skeletons with three links are
\begin{align}
S_{5,17} \qquad\qquad&& S_{5,18}\qquad\qquad \nonumber\\
\begin{tikzpicture}
\node[] (1) [label=below: $m_5$]{$\circ$};
\node[] (2) [above right =3mm and 5mm of 1, label =above: $m_4$]{$\circ$};
\node[] (3) [below right = 3mm and 5mm of 1, label = below: $m_3$]{$\circ$};
\node[] (4) [right = 2 cm of 1, label = below: $m_2$]{$\circ$};
\node[] (5) [right = of 4, label = below: $m_1$]{$\circ$};
\draw[->] (1) -- (2);
\draw[->] (2) -- (4);
\draw[->] (3) -- (4);
\draw[->] (4) -- (5);
\draw[->] (5) to [out = 135, in = 45] (4);
\draw[dashed] (2) --(3);
\draw[dashed] (3) to [out = -45, in = -135] (5);
\draw[dashed] (2) to [out = 45, in = 90] (5);
\end{tikzpicture}
&&
\begin{tikzpicture}
\node[] (1) [label=above: $m_4$]{$\circ$};
\node[] (2) [below =7mm of 1, label =left: $m_1$]{$\circ$};
\node[] (3) [below left = 5mm and 5mm of 2, label = below: $m_2$]{$\circ$};
\node[] (4) [below right = 5 mm and 5mm of 2, label = below: $m_3$]{$\circ$};
\node[] (5) [right = 5mm of 2, label = above: $m_5$]{$\circ$};
\draw[->] (1) -- (2);
\draw[->] (2) -- (3);
\draw[->] (3) -- (4);
\draw[->] (4) -- (2);
\draw[->] (5) -- (2);
\draw[dashed] (1) --(5);
\draw[dashed] (4) -- (5);
\draw[dashed] (1) to [out = 0, in = 0] (4);
\end{tikzpicture}
\end{align}
Finally,  we have the skeletons with $4$ and $6$ links:
\begin{align}
S_{5,19}\qquad\qquad && S_{5,20}\qquad\qquad \nonumber\\
\begin{tikzpicture}[baseline=-28mm]
\node[] (1) [label=above: $m_5$]{$\circ$};
\node[] (2) [below = of 1, label =left: $m_4$]{$\circ$};
\node[] (3) [ right =  of 2, label = below: $m_1$]{$\circ$};
\node[] (4) [right = of 3, label = below: $m_2$]{$\circ$};
\node[] (5) [right = of 4, label = below: $m_3$]{$\circ$};
\draw[->] (1) -- (3);
\draw[->] (2) -- (3);
\draw[->] (3) -- (4);
\draw[->] (4) to [out = 135, in = 45] (3);
\draw[->] (5) -- (4);
\draw[dashed] (1) --(2);
\draw[dashed] (4) to [out = 90, in = 0] (1);
\draw[dashed] (2) to [out = -45, in = -135] (4);
\draw[dashed] (5) to [out = -135, in = -45] (3);
\end{tikzpicture} &&
\begin{tikzpicture}
\node[] (1) [label=above: $m_1$]{$\circ$};
\node[] (5) [below = of 1, label = below left: $m_2$]{$\circ$};
\node[] (2) [left = of 5, label =left: $m_5$]{$\circ$};
\node[] (4) [ right =  of 5, label = below: $m_3$]{$\circ$};
\node[] (3) [below = of 5, label = below: $m_4$]{$\circ$};
\draw[->] (1) -- (5);
\draw[->] (2) -- (5);
\draw[->] (3) -- (5);
\draw[->] (4) -- (5);
\draw[->] (5) to [out = 45, in = -45] (1);
\draw[dashed] (1) --(2);
\draw[dashed] (4) -- (1);
\draw[dashed] (3) -- (4);
\draw[dashed] (2)--  (3);
\draw[dashed] (1) to [out = 0 , in = 0] (3);
\draw[dashed] (2) to [out = 45 , in = 135] (4);
\end{tikzpicture}
\end{align}

\subsection{A priori bounds on exponents}
To produce a LG theory each of the skeletons above is assigned a set of exponents, i.e. the integers $m_i$ for each root or pointer variable.  Since we are not interested in variables with $q_i = 1/2$, $m_i \ge 3$ when $X_i$ is a root, and $m_i \ge 2$ when $X_i$ is a pointer.  For $n < 3$ these are the only constraints on the $m_i$.  Hence, each superpotential in~(\ref{eq:n2LG}) yields a gc family with $c<6$.

For $n \ge 3$ there are non-trivial bounds on $\qtot$ given in~(\ref{eq:qtotbounds}), and these translate into bounds on the exponents $m_i$.  In order to avoid drowning in a case-by-case analysis, we would like to start with some a priori bounds on the $m_i$ that do not depend on the details of the skeleton.  Such a bound follows from theorem~\ref{thm:chargebound}.  The details are relegated to the Appendix, and here we just state the results.  

Let  $\bM = (m_1,m_2,\ldots,m_n)$ be a set of exponents ordered according to 
$$m_1\ge m_2 \ge \cdots \ge m_n~.$$ 
A necessary condition for~(\ref{eq:qtotbounds}) is that
\begin{align}
\bM \in \Sigma_n = \Sigma_{n}^{\text{spor}} \cup \Sigma_n^{\infty}~,
\end{align}
where $\Sigma_n^{\infty}$ is composed is composed of infinite families of ordered $n$-tuples of integers, while $\Sigma_n^{\text{spor}}$ is a finite list of ordered $n$-tuples of integers.  To list these we will use a short-hand: e.g. $(6..7,6,5)$ will stand for $(6,6,5),(7,6,5)$~.  For $n=3$ we find
\begin{align}
\Sigma_3^{\text{spor}} & = 
\{
(6..7,6,5)~,~~~~~
(5..9,5,5)~,~~~~~
(7..9,7,4)~,~
(6..11,6,4)~,
(5..19,5,4)~,\nonumber\\
&\qquad (11..13,11,3)~, 
(10..14,10,3)~,
(9..17,9,3)~,
(8..23,8,3)~,
(7..41,7,3)
\}~, \nonumber\\
\Sigma_3^{\infty} & = 
\{
(4..\infty,4,4)~,
(3..\infty,3,3)~,
(4..\infty,4,3)~,
(5..\infty,5,3)~,
(6..\infty,6,3)\} \nonumber\\
&\qquad
\cup \bigcup_{k\ge 2}
\{(k..\infty,k,2)
\}
\end{align}
The $n=4,5$ the results are
\begin{align}
\Sigma_4^{\text{spor}} & = \{(5,4,4,3)~,(4,4,4,3)~,(5..7, 5,3,3)~,(4..11,4,3,3)\} \nonumber\\
& \qquad \cup \bigcup_{\bM \in \Sigma_3^{\text{spor}}} \{(\bM,2) \}~,\nonumber\\
\Sigma_4^{\infty} & = \{ (3..\infty,3,3,3)\}\cup \bigcup_{\bM \in \Sigma_3^{\infty}} \{(m_1,m_2,m_3,2) \}~, \nonumber\\
\Sigma_5^{\text{spor}} & = \{(3..5,3,3,3,3)\}\cup \bigcup_{\bM \in \Sigma_4^{\text{spor}}} \{(\bM,2) \}~, \nonumber\\
\Sigma_5^{\infty} & = \bigcup_{\bM \in \Sigma_4^{\infty}} \{(\bM,2)\}~.
\end{align}
We are finally ready to tackle our graphs.  In the remainder of this section we will describe the feasible irreducible graphs.  Since $n\le 2$ is trivial, we begin with the $n=3$ skeletons.

\subsection{Exponents for $n=3$ graphs}

\subsubsection*{\underline{$S_{3,1}$}}
Here there are no links, and we just need to describe the bounds on the exponents that lead to $c< 6$, i.e. $\qtot >1/2$.  Our first step is to use~(\ref{cor:n3}) to limit the possible infinite families of exponents.
Running through the finite list of cases in $\Sigma^\infty_3$, we find the following infinite families:
\begin{align}
\label{eq:infinitefamiliesS31}
(2,\ast,\ast)~, 		&& (\ast,2,\ast)~, && \nonumber\\
(3,3,\ast)~,		&& (3,\ast,3)~,	&& (\ast,3,3)~, \nonumber\\
(3,4,\ast)~,		&& (3,\ast,4)~,	&& (\ast,3,4)~, \nonumber\\
(4,3,\ast)~,		&& (4,\ast,3)~,	&& (\ast,4,3)~, \nonumber\\
(3,\ast,5)~,		&& (3,\ast,6)~,	&& (4,\ast,4)~, \nonumber\\
(5,\ast,3)~,		&& (6,\ast,3)~.
\end{align}
Eliminating these from $\Sigma_{3}$ leads to a finite list of exponents: we find $220$ possibilities with a maximum exponent value of $25$.\footnote{To spare the reader we will not list these $220$ possibilities explicitly --- they are easy to recover given the bound on the maximum exponent.  When quoting the list of sporadic possibilities is similarly impractical, we will follow this procedure in the cases below.  We will also provide an online list of the charges for all irreducible theories that are not contained in one of the explicitly listed infinite families.}

Not every set of exponents leads to an irreducible LG family.  For instance, the family with exponents $\bM = (m_1,4,3)$ is deformation--equivalent to
\begin{align}
W = X_3^3 + X_2^6 + X_1^{m_1} X_2~.
\end{align}
There are many similar decompositions, but there are also irreducible families; as an example, consider $\bM = (m_1,3,3)$, with $m_1$ not divisible by $7$.  The most general superpotential is
\begin{align}
W = a_1 X_1^{m_1} X_2 + a_2 X_2^3 X_3 + a_3 X_3^3~,
\end{align}
and this is irreducible.  

\subsubsection*{\underline{$S_{3,2}$}}
This diagram is symmetric in the $1,3$ indices, so we will assume $m_3 \ge m_1$.  In fact, we can assume $m_3 >m_1$ because if $m_3=m_1$ and the link is feasible, then both $X_3$ and $X_1$ can be roots, i.e. the graph will be reducible.  Since $X_2$ is a root we also have $m_2 >2$.  With these restrictions 
 we find from~(\ref{cor:n3}) that the infinite families satisfying $\qtot >1/2$ are
\begin{align}
(2,\ast,\ast) &&
(3,3,\ast)	 &&
(3,4,\ast)	&&
(4,3,\ast)   &&
(3,\ast,4) &&
(3,\ast,5) &&
(3,\ast,6)~.
\end{align}
The last three cases are not feasible if $m_2$ is larger than $16$.  To see this, consider the exponent $\bM=(3,m_2,5)$.  The charges in this case are $(\ff{m_2-1}{3 m_2}, \ff{1}{m_2}, \ff{m_2-1}{5m_2})$.  Thus, to have the link between $1$ and $3$ we need a solution to
\begin{align}
5p_1 + 3p_3 = \frac{15 m_2}{m_2-1} = 15 + \frac{15}{m_2-1}~.
\end{align}
So, $m_2 \le 16$.  Proceeding in the same fashion, we find stronger bounds on $m_2$ for the exponents $\bM = (3,\ast,4)$ and $\bM = (3,\ast ,6)$.

Next, we consider the cases $(4,3,\ast)$, $(3,4,\ast)$, and $(3,3,\ast)$.
\begin{enumerate}
\item $\bM = (4,3,m_3)$.  In this case the charges are $(\frac{1}{6},\frac{1}{3},\frac{2}{3m_3})$.  Since $X_1$ and $X_2$ can both be roots, this is feasible; on the other hand, it is manifestly a reducible LG family, so we need not consider it further.
\item $\bM = (3,4,m_3)$.  The charges are $(\frac{1}{4},\frac{1}{4},\frac{3}{4m_3})$, and this is again a reducible LG family.
\item $\bM = (3,3,m_3)$.  The charges are $(\frac{2}{9},\frac{1}{3},\frac{2}{3m_3})$.  The link requires
\begin{align}
2 m_3 p_1 + 6 p_3 = 9 m_3~.
\end{align}
This is possible only if $m_3 = 2 k$, in which case $q_3 = \frac{1}{3k}$.  Thus, the link is feasible if and only if the LG family is reducible.
\end{enumerate}
The final case, $\bM = (2,m_2,m_3)$ takes a little work following Diophantus.  The feasibility condition is
\begin{align}
(m_2-1) \left[ m_3 p_1 + 2 p_3\right] = 2m_2 m_3~.
\end{align}
Since $\gcd(m_2-1,m_2) = 1$ solutions are possible only if there exists an integer $\ell$ such that
\begin{align}
(m_2 - 1) \ell &= 2 m_3~,&
m_3 p_1 + 2p_3 & = m_2 \ell~.
\end{align}
If $m_2$ is even, then the first equation implies $\ell = 2 r$, but  then $q_3 = \frac{1}{2kr}$.  So, this possibility leads to a reducible family.  

If, on the other hand, $m_2 = 2k +1$, then $m_3 = k\ell$, and the second equation becomes
\begin{align}
2 p_3 = \ell (1 + k (2-p_1) )~.
\end{align}
If $p_1 =0$ or $p_3 = 0$, then the family will be reducible.  So, we need only consider $p_1 = 1$ or $p_1 = 2$.  If $p_1 = 2$, then $\ell = 2p_3$, and thus $m_3 = 2k p_3$.  This implies $q_3 = \frac{1}{p_3(2k+1)}$, so we again find a reducible family.  The last interesting possibility is $p_1 = 1$.  If $\ell$ is even, we again find a reducible family, but if $\ell = (2s+1)$, we find an interesting possibility.  In that case a solution for $p_3$ is possible if and only if $k = 2t+1$, so that we find a family
\begin{align}
\bM = (2, 4t+3, (2s+1)(2t+1) )~.
\end{align}
This is the only infinite family of exponents that yields irreducible feasible graphs for $S_{3,2}$.  Eliminating this infinite family from $\Sigma_3$, we find just three remaining irreducible exponents:
\begin{align}
(3, 7, 4)~, && (3, 13, 4)~, && (3, 16, 5)~.
\end{align}

\subsubsection*{\underline{$S_{3,3}$}}
The graph has cyclic symmetry.  Using~(\ref{cor:n3}) we find that up to cyclic symmetry the infinite families are
\begin{align}
(2,\ast,\ast)~, &&
(3,3,\ast)~, 	&&
(3,4,\ast)~,	&&
 (4,3,\ast)~.
\end{align}
In addition to these $\Sigma_3$ contains, up to cyclic permutations, $48$ sets of exponents with maximum exponent of $18$.

\subsubsection*{\underline{$S_{3,4}$}}
By similar manipulations, we find that all irreducible exponents are contained in three infinite families
\begin{align}
\bM &= ( 2t(2s+1)+1,4t+1,2)~, &
\bM & = (4s+2,3,4)~,&
\bM &= (3s-1,4,3)~, 
\end{align}
or the short list
\begin{align}
\bM \in \{ (3,6,5)~,~ (5,7,3)~,~(3, 8, 7)~,~(4, 6, 5)~,~(5, 5, 4)~,~ (7, 7, 3)~,~ (7, 10, 3)~,~ (8, 8, 3)\}~.
\end{align}
Note that the finite list only contains three distinct charges:
\begin{align}
q &= (\ff{2}{23},\ff{3}{23},\ff{7}{23})~,&
q & = (\ff{3}{23},\ff{4}{23},\ff{5}{23})~,&
q & = (\ff{2}{17},\ff{3}{17},\ff{5}{17})~.
\end{align}
This completes the analysis of exponents for $n=3$ graphs without pointing links.

\subsection{Exponents for the $n=4$ graphs}
We now turn to $n=4$.  Although the number of skeletons is large, we can reuse the $n=3$ results because most of the $n=4$ graphs arise as \textit{decorations} of $n=3$ graphs.  In other words, they have one of two forms:
\begin{align}
\begin{tikzpicture}
\draw (0,0) node[label=below: $m_4$] {$\circ$};
\draw (3,0) circle[radius=0.7];
\draw (3,0) node {$G_3$};
\draw[->] (0.5,0)--(1.8,0);
\end{tikzpicture}
&&
\begin{tikzpicture}
\draw (0,0) node[label=below: $m_4$] {$\circ$};
\draw (3,0) circle[radius=0.7];
\draw (3,0) node {$G_3$};
\draw[->] (0.5,0)--(1.8,0);
\draw[dashed] (0.2,0.2) to[out = 30, in = 170] (2.4,0.55);
\end{tikzpicture}
\end{align}
Thus, we have\footnote{Recall that $\pi (i)$ denotes the node pointed to by node $i$.}
\begin{align}
q_4 & = \frac{1}{m_4} \left(1-q_{\pi(4)}\right)~, &
\qtot(G_4) & = \qtot(G_3) + q_4~.
\end{align}
Since the pointing exponent $m_4 \ge 2$, we conclude that
\begin{align}
\qtot(G_4) -1 \le \qtot(G_3) -\frac{1}{2}~.
\end{align}
If $G_4$ is a decoration of $G_3$, then to find the exponents on $G_4$ consistent with~(\ref{eq:qtotbounds}), i.e. those consistent with $c<6$ we can restrict attention to $G_3$ that satisfy the bound for $n=3$.  Moreover, whenever $G_3$ is reducible, then $G_4$ is necessarily reducible as well:  this means that we need only worry about decorating the irreducible $n=3$ graphs.  Finally, suppose that 
\begin{align}
\max_{\text{irred}~~G_3}  \qtot(G_3) = q_\ast < 1~.
\end{align}
Then for every irreducible $G_4$ that is a decoration of $G_3$ the $m_4$ exponent is bounded.  We will now apply these ideas to the $n=4$ graphs.

\subsubsection*{\underline{$S_{4,1}$}}
$S_{4,1}$ is a decoration of $S_{3,1}$ with $X_4$ pointing at $X_1$, and therefore
\begin{align}
q_4 = \frac{1}{m_4} \left(1-q_1\right)~.
\end{align}
Since $\qtot$ is monotonically decreasing in each exponent, it is easy to see that $\qtot(S_{3,1}) \le 1$
for all exponents, with equality precisely for $\bM = (2,2,3)$; however, this has $q_1 = q_2 = q_3 = 1/3$ and obviously decomposes.  Hence, for irreducible $S_{3,1}$ we have a better bound
\begin{align}
\qtot(S_{3,1} ) \le \frac{17}{18}~,
\end{align}
which is realized by $\bM = (2,3,3)$.  Thus, for irreducible $S_{4,1}$ we have the bound
\begin{align}
\qtot(S_{4,1})-1 \le \qtot(S_{3,1} ; (2,3,3)) + q_4-1 = \frac{1}{18}\left(-1 + \frac{11}{m_4}\right)~.
\end{align}
This means that $2\le m_4 \le 10$.  We now apply this to the $S_{3,1}$ exponents in~(\ref{eq:infinitefamiliesS31}) and find
the following infinite families of exponents that contain irreducible theories:
\begin{align}
(\ast,2,\ast,2)~, 		&&(2,3,\ast,2)~,		&&(2,\ast,3,2)~,		&&(2,\ast,4,2)~, \nonumber\\
(2,2,\ast,3)~,		&&(2,\ast,3,3)~,		&&(\ast,2,3,3)~,		&&(3,\ast,3,2)~,\nonumber\\
(\ast,3,3,2)~,		&&(\ast,3,4,2)~,		&&(\ast,4,3,2)~.
\end{align}
Once we eliminate these from $\Sigma_4$, we find there are an additional $108$ irreducible exponents, with maximum exponent of $13$.

\subsubsection*{\underline{$S_{4,2}$}}
$S_{4,2}$ is a decoration of $S_{3,2}$, with $X_4$ pointing at $X_3$.  Thus, there are just two possible exponents for an infinite family of irreducible theories:
\begin{align}
\bM &= (2,4t+3,(2s+1)(2t+1),m_4)~,&
\bM &= ((2s+1)(2t+1),4t+3,2,m_4)~.
\end{align}
We also have $\qtot(S_{3,2}) \le 6/7$ for irreducible theories, with equality realized by $\bM = (2,7,3)$ or $\bM = (3,7,2)$.  This immediately leads to a bound $m_4 \le 4$.  Taking a look at $\qtot(S_{4,2})$, we then obtain just one infinite family with irreducible theories:
\begin{align}
\bM = (2,4t+3,(2s+1)(2t+1),2)~.
\end{align}
Eliminating this from $\Sigma_4$, we find that there are  $5$ additional irreducible theories with exponents:
\begin{align}
(2, 7, 3, 3)~,&& (2, 7, 3, 4)~,&& (3, 7, 2, 2)~, && (3, 7, 2, 3)~, && (3, 7, 4, 2)~.
\end{align}

\subsubsection*{\underline{$S_{4,3}$}}
$S_{4,3}$ decorates $S_{3,1}$ with $X_4$ pointing at $X_2$ and having a link to $X_1$.  Using similar arguments to those for $S_{4,1}$, we find that $m_4 \le 13$ for irreducible theories.  We then check all the infinite families in $S_{3,1}$ with $m_4$ satisfying the bound, and find a number of subsets of unbounded exponents.  These include, for instance, $\bM = (3,m_2,3,2)$ and $\bM=(2,m_2,m_3,2)$.  However, when we consider the necessary link between $X_4$ and $X_1$, we find that there are, in fact, no unbounded irreducible families.  For instance, in the case of $\bM =(2,m_2,m_3,2)$ $q_1 = q_4$, so that if a link is possible, then $X_1$ and $X_4$ can both be roots.  On the other hand, for $\bM = (3,m_2,3,2)$, we find that the link requires a solution to 
\begin{align}
2 p_1 + 3 p_4 = \frac{18 m_2}{3 m_2 -2} = 6 + \frac{12}{3m_2-2}~.
\end{align}
Since the left-hand-side must be an integer, $m_2$ is bounded.\footnote{In fact, here there are no $m_2\ge2$ for which the right-hand-side is an integer, but the point is that even if there were non-trivial solutions, they would be bounded.}  After checking through every infinite family of exponents in $\Sigma_3^{\infty}$, we find that every exponent yields a graph that is reducible whenever it is feasible.  We then also check that the same holds for every $\bM \in\Sigma_3^{\text{spor}}$.  Thus, $S_{4,3}$ contributes no new irreducible theories.  This has a nice consequence for the $n=5$ analysis:  every $n=5$ graph that is a decoration of $S_{4,3}$ can now be ignored.

\subsubsection*{\underline{$S_{4,4}$, $S_{4,6}$, $S_{4,8}$,  $S_{4,9}$}}
All of these graphs behave as $S_{4,3}$:  there are no irreducible theories with $c<6$.

\subsubsection*{\underline{$S_{4,5}$}}
This graph is not a decoration of an $n=3$ graph.  We use~(\ref{cor:n4}) to restrict the possible unbounded exponents, and we find that up to cyclic symmetry there are just three unbounded families:
\begin{align}
\bM &= (2,\ast,2,\ast)~,&
\bM &=(2,2,3,\ast)~,&
\bM &=(3,2,2,\ast)~.
\end{align}
$\Sigma_3$ contains $22$ more sets of exponents with a maximum exponent of $8$.

\subsubsection*{\underline{$S_{4,7}$}}
This is a decoration of $S_{3,3}$, and there is an infinite family of irreducible theories with exponents
\begin{align}
\bM = (2,2,3k+1,3)~.
\end{align}
Removing this family from $\Sigma_4$, we find there is just one more irreducible exponent:
\begin{align}
\bM = (3,2,2,5) ~.
\end{align}

\subsection{Exponents for the $n=5$ graphs}

There is nothing conceptually different about the analysis of the $n=5$ cases, but it does become a bit more tedious because there are more possibilities for the potentially infinite families of exponents.  To handle these possibilities, we made a scan through all exponents with $\max \bM \le 100$ for irreducible feasible graphs with $\qtot >3/2$.  We found no evidence of  infinite families and  in fact just a short list of possible exponents with $\max \bM = 10$.  We believe we have the full set of exponents, and we will discuss this further below.  For now we will simply summarize the results of our scan through the exponents.

Eleven of the twenty $n=5$ graphs are decorations of $n=4$ graphs that are reducible for every exponent.  For instance, $S_{5,5}$ is a decoration of $S_{4,3}$, with $X_5$ pointing at $X_1$.  Following the general discussion of decorations, we need not consider such $n=5$ graphs.  This leaves us to consider nine graphs: $S_{5,\alpha}$ with $\alpha \in \{1,2,3,4,10,11,12,15,18\}$.  Of these only $S_{5,1}$, $S_{5,2}$ and $S_{5,10}$ admit irreducible theories, and, furthermore, the single irreducible exponent for $S_{5,2}$ produces a set of charges already contained in the exponents of $S_{5,1}$.  The full list of exponents is then as follows:  $S_{5,1}$ has $19$ exponents
\begin{align*}
(2, 2, 4, 2, 2)~,&& (2, 2, 4, 3, 2)~,&& (2, 2, 4, 4, 2)~,&& (2, 2, 4, 5, 2)~,&& (2, 2, 5, 2, 2)~, \nonumber\\
(2, 2, 5, 3, 2)~,&& (2, 2, 6, 2, 2)~,&& (2, 2, 7, 2, 2)~,&& (2, 2, 8, 2, 2)~,&& (2, 2, 9, 2, 2)~, \nonumber\\
(2, 2, 10, 2, 2)~,&& (2, 3, 3, 2, 2)~,&& (2, 3, 3, 3, 2)~,&& (2, 3, 3, 4, 2)~,&& (2, 3, 3, 5, 2)~,\nonumber\\
(2, 5, 3, 2, 2)~,&& (2, 7, 3, 2, 2)~,&& (2, 9, 3, 2, 2)~,&& (3, 2, 4, 2, 2)~,
\end{align*}
and $S_{5,10}$ has $7$ exponents (up to cyclic symmetry):
\begin{align*}
(2, 2, 2, 2, 3)~, &&(2, 2, 2, 2, 4)~, &&(2, 2, 2, 2, 5)~, &&(2, 2, 2, 2, 6)~, \nonumber\\
 (2, 2, 2, 2, 7)~, &&(2, 2, 2, 3, 3)~, &&(2, 2, 3, 2, 3)~.
\end{align*}

\section{Pointing links} \label{s:pointing}
We now turn to the final aspect of the classification:  the analysis of pointing links.  Again, in principle the problem is straightforward:  for every skeleton in $5$ variables or less, we are to consider every possible way of making every link point at some node.  The only requirement modulo feasibility of the link is that a link between variables $X_i,X_j$ cannot point at $X_i$, $X_j$ or any $X_k$ that is a target of a pointer.  The last point may be a little bit surprising, so we illustrate it with a simple example based on $S_{5,2}$.  Suppose the link between $X_1$ and $X_3$ points at the variable $X_4$, so that a representative superpotential is
\begin{align}
W = (X_1^2+X_3^2)X_2 + X_2^3 + X_4^2X_1 + X_5^2X_4 + X_1X_3X_4~.
\end{align}
In this case $dW = 0$ if $X_2=X_4=0$ and
\begin{align}
X_1^2+X_3^2 &= 0~,&
X_1X_3+X_5^2 & = 0~,
\end{align}
and this has solutions with $X_1,X_3,X_5$ non-zero.

\subsection{Pointing links for disconnected skeletons}
The most interesting feature of a pointing link is that it can connect graphs with disconnected skeletons;  we already saw examples in section~\ref{ss:pointedfeasibility}.

Let us investigate this more systematically by examining a graph with disconnected skeleton $S_{3,2}+S_{1,1}$. The only way this can be irreducible is if the link of $S_{3,2}$ points at the node of $S_{1,1}$:
\begin{align}
\begin{tikzpicture}
\node [] (1)   [label=below:$m_1$] {$\circ$};
\node [] (2)  [right=of  1, label=below:$m_2$] {$\circ$};
\node [] (3) [right=of 2, label = below: $m_3$] {$\circ$};
\node [] (L1) [above=2mm of 2] {$\bullet$};
\node [] (4) [above= 5mm of L1, label = above :$m_4$] {$\circ$};
\draw [->]  (3) -- (2);
\draw [->] (1) -- (2);
\draw[dashed] (L1) to [out = 0, in =90] (3);
\draw[dashed] (L1) to [out = 180, in = 90] (1);
\draw[dashed] (L1) to [out = 90, in = -90] (4);
\end{tikzpicture}
\end{align}
We examine the exponents in $\Sigma^{\infty}_4$ and find a large number of infinite families consistent with $\qtot >1$.  However, in all but one case, a solution to the pointing link only has a finite number of solutions.  For instance, we have the family $\bM = (2,\ast,3,5)$, for which the link requires
\begin{align}
3 p_1 + 2p_2 = \frac{24 k}{5k-5}~,
\end{align}
and this has a unique solution for $k=25$.  For all but one other case, the feasibility condition always takes this form of
\begin{align}
a p_1 + b p_2 = \frac{ c k}{d(k-1)}~,
\end{align}
for $a,b,c,d$ some positive integers:  in every one of these situations, the possible values of $k$ are therefore bounded.

The exception is the exponent $\bM = (2,k,2,l)$, for which the charges are
\begin{align}
q_1 &= q_3 = \frac{k-1}{2k}~,& 
q_2 & = \frac{1}{k}~,&
q_4 & = \frac{1}{l}~.
\end{align}
This is a reducible family:  since $q_1 = q_3$, if a link is feasible, then $X_1$ can point at $X_4$, so that we simply obtain the disconnected skeleton $S_{2,1}+S_{2,1}$.  It follows that only a finite number of exponents yields feasible and irreducible theories.

\subsection{General analysis via numeric study}
While the analysis of the general case is in principle straight-forward, the number of cases to check grows rather large as we increase $n$ and the number of pointing links.  A link typically has more than one possible target, and we must consider all possibilities.  It is easy to scan through $\Sigma_n^{\text{spor}}$ for all of these possibilities, but the existence of the potentially infinite families in $\Sigma_n^{\infty}$ presents a challenge.   In every graph that we checked manually, we were able to rule out irreducible theories for all but a finite subset of $\Sigma_n^{\infty}$.  As we illustrated in the examples above, this sort of result typically relies on three factors:  the bound on $\qtot$, the necessary conditions for irreducibility, and the feasibility of links.  Typically these factors are intertwined, and we have not been able to distill the constraints into a rigorous finiteness result.

Instead, we made a numeric investigation of all pointing links for exponents bounded by $\max \bM \le 100$.\footnote{It turns out that we never need to consider links between links: even without the constraints from links between links, there is no feasible graph where two pointing links have the same target.}   We found no evidence of infinite families of exponents that contain irreducible theories and in fact a rather small list of distinct charges.  The connected $n=4$ graphs with at least one pointing link yielded $20$ distinct charges from the following skeletons:  $S_{3,4}+S_{1,1}$, $S_{3,2}+S_{1,1}$, $S_{4,2}$, and $S_{4,4}$.  These realized a maximum exponent of $31$, while the $n=5$ graphs with at least one pointing link produced $2$ distinct charges from the skeleton $S_{4,2}+S_{1,1}$.

These results, together with investigations of infinite families for specific graphs, amount to substantial evidence that we do indeed have a complete list of irreducible theories, but it would certainly be useful to have a rigorous proof of the assertion --- such a proof would likely yield insights into LG combinatorics.  We leave this for future work.

\section{Summary of Results}
We studied the combinatorics of LG theories.  Our work of course has a close relation to the classification of quasi--homogeneous singularities that has been of substantial interest to mathematicians, and to our surprise we found some singularities that were overlooked in previous work.  The physics perspective offers some useful simplifications, most of which were indeed anticipated in~\cite{Vafa:1988uu}:  it is natural to organize the theories by $n$, and to restrict attention to indecomposable theories; moreover, instead of organizing theories at fixed Milnor number (in the SCFT this is the number of relevant and marginal deformations of the theory), it is natural to focus on the central charge.  We applied this philosophy to the problem of finding all indecomposable theories with $c<6$.

Let us now summarize our results.  We will describe the set of irreducible theories for various $1\le n \le 5$.  We remind the reader that the results for $n=4,5$ do rely on some numeric studies.  Two more points should be borne in mind:  first, while we claim to have every indecomposable theory, our lists do contain some decomposable theories; second, because of the infinite families present some indecomposable theories may occur more than once.  We claim, however, that every indecomposable theory is included.

In what follows we present the infinite families of charges.  The sporadic cases are provided in the supplementary files included with the ArXiv submission.

\subsection{rigorous results for $n\le 3$}

\subsubsection*{\underline{ $n=1$}} 
We have the familiar list of $c<3$ $A_{k-2}$ minimal models with charge $q=1/k$.
\subsubsection*{\underline{$n=2$}}
There are two sets of charges:
\begin{align}
 \left(\frac{l -1}{kl}~,  \frac{1}{l} \right)~,&&
 \left(\frac{l -1}{kl-1}~,\frac{k-1}{kl-1}\right)~.
\end{align}
\subsubsection*{\underline{$n=3$}}
$S_{3,1}$ yields the infinite families
\begin{align}
\left(\frac{7}{9k}, \frac{2}{9}, \frac{1}{3} \right)~, &&
\left(\frac{2k+1}{9k}, \frac{k-1}{3k}, \frac{1}{k}\right)~,&&
\left(\frac{2k+1}{12k}, \frac{k-1}{3k}, \frac{1}{k}\right)~,&&
\left(\frac{3k-2}{9k}, \frac{2}{3k}, \frac{1}{3}\right) \nonumber\\
\left(\frac{3k-2}{12k}, \frac{2}{3k}, \frac{1}{3}\right)~,&&
\left(\frac{3k-2}{15k}, \frac{2}{3k}, \frac{1}{3}\right)~,&&
\left(\frac{3k-2}{18k}, \frac{2}{3k}, \frac{1}{3}\right)~,&&
\left(\frac{3k+1}{12k}, \frac{k-1}{4k}, \frac{1}{k}\right)~,\nonumber\\
\left(\frac{4k-3}{12k}, \frac{3}{4k}, \frac{1}{4}\right)~,&&
\left(\frac{4k-3}{16k}, \frac{3}{4k}, \frac{1}{4}\right)~,&&
\left(\frac{5k-4}{15k}, \frac{4}{5k}, \frac{1}{5}\right)~,&&
\left(\frac{6k-5}{18k}, \frac{5}{6k}, \frac{1}{6}\right)~,&&\nonumber\\
\left( \frac{l+1}{2lk}, \frac{l-1}{2l}, \frac{1}{l}\right)~,&&
\left( \frac{kl-l+1}{2lk}, \frac{l-1}{lk}, \frac{1}{l}\right)~.
\end{align}
$S_{3,2}$ yields the infinite family
\begin{align}
\left(\frac{2k+1}{4k+3}, \frac{1}{4k+3}, \frac{2}{(4k+3)(2l+1)}\right)~.
\end{align}
$S_{3,3}$ yields the infinite families
\begin{align}
\left(\frac{2k+1}{9k+1},\frac{3k-2}{9k+1},\frac{7}{9k+1}\right)~,&&
\left(\frac{2k+1}{12k+1},\frac{4k-3}{12k+1},\frac{10}{12k+1}\right)~,\nonumber\\
\left(\frac{3k+1}{12k+1},\frac{3k-2}{12k+1},\frac{9}{12k+1}\right)~,&&
\left(\frac{kl-l+1}{2kl+1},\frac{2l-1}{2kl+1},\frac{k+1}{2kl+1}\right)~.
\end{align}
To describe the $S_{3,4}$ families, set $P(k,l) = 1+8k(2kl+k+l+1)$.  Then the charges are
\begin{align}
\left(\frac{4k+1}{P(k,l)},\frac{2k(2l+1)}{P(k,l)},\frac{(4k+1)(4kl+2k+1)}{2P(k,l)}\right)~,&&
\left(\frac{3}{12k-5},\frac{3k-2}{12k-5},\frac{3k-1}{12k-5}\right)~,\nonumber\\
\left(\frac{2}{12k+5},\frac{4k+1}{12k+5},\frac{2k+1}{12k+5}\right)~.
\end{align}
These are the complete infinite families for $n\le 3$.

\subsection{Infinite families for $n =4,5$}
Our numerics suggest that there are no infinite families of irreducible theories for any graphs with pointing links or $n=5$, but as we emphasized, we have not a rigorous proof of this.  However, for $n=4$ and no pointing links the results are rigorous.

$S_{4,1}$ yields the infinite families
\begin{align}
\left(\frac{7}{9k}, \frac{2}{9}, \frac{1}{3}, \frac{9k-7}{18k}\right)~,&&
\left(\frac{k+1}{4k}, \frac{k-1}{2k}, \frac{1}{k}, \frac{3k-1}{12k}\right)~,&&
\left(\frac{2k+1}{6k}, \frac{k-1}{3k}, \frac{1}{k}, \frac{4k-1}{12k}\right)~,\nonumber\\
\left(\frac{3k-2}{6k}, \frac{2}{3k}, \frac{1}{3}, \frac{3k+2}{12k}\right)~,&&
\left(\frac{3k-2}{6k}, \frac{2}{3k}, \frac{1}{3}, \frac{3k+2}{18k}\right)~,&&
\left(\frac{3k-2}{9k}, \frac{2}{3k}, \frac{1}{3}, \frac{3k+1}{9k}\right)~,\nonumber\\
\left(\frac{4k-3}{8k}, \frac{3}{4k}, \frac{1}{4}, \frac{4k+3}{16k}\right)~,&&
\left(\frac{l+1}{2kl}, \frac{l-1}{2l}, \frac{1}{l}, \frac{2kl-l-1}{4kl}\right)~.
\end{align}

$S_{4,2}$ yields one infinite family:
\begin{align}
\left(  \frac{1}{2} \frac{4k + 2}{4k + 3},   \frac{1}{4k+3},   \frac{4k+2}{(2l+1)(2k+1)(4k+3)},   \frac{1}{2} \frac{ (4k+3)(2l+1)(2k+1)-4k-2}{(2l+1)(2k+1)(4k+3) }  \right)
\end{align}

$S_{4,5}$ yields the infinite families
\begin{align}
\left(\frac{3k-1}{12k-1}, \frac{3k+2}{12k-1}, \frac{6k-5}{12k-1}, \frac{9}{12k-1}\right)~,&&
\left(\frac{4k-1}{12k-1}, \frac{4k+1}{12k-1}, \frac{4k-3}{12k-1}, \frac{8}{12k-1}\right)~,\nonumber\\
\left(\frac{2kl-l-1}{4kl-1}, \frac{2l+1}{4kl-1}, \frac{2kl-k-1}{4kl-1}, \frac{2k+1}{4kl-1}\right)~.
\end{align}
Finally,  $S_{4,7}$ has the infinite family
\begin{align}
\left(\frac{3k+2}{12k+5}, \frac{6k+1}{12k+5}, \frac{3}{12k+5}, \frac{3k+1}{12k+5}\right)~.
\end{align}

\subsection{The sporadic cases}
The are also $418$ distinct charges that do not fit into the above infinite families:  $236$ for $n=3$; $154$ for $n=4$, and merely $28$ for $n=5$.  For obvious reasons we will not list these here, although the reader will find many in the examples given above.  Instead, we will include them in the ArXiv submission.  We will list the charges as follows:  for each distinct ordered set $(q_1,q_2,\ldots, q_n)$ we will write the $q_i = a_i/b$, where $b$ is a common denominator, and we will list the integers $(a_1,a_2,\ldots,a_n,b)$.  These charges are bounded by $q_i \ge \frac{1}{31}$.

\appendix

\section{Some a priori integer bounds}
In this appendix we solve some combinatoric inequalities that arise in the classification of skeletons and exponents with $c<6$.  These results are of course not new for $n\le 4$ --- see, e.g.~\cite{Xie:2015rpa}.  As we will see, for $n=5$ there are in fact just two more cases that need to be considered in addition to the $n\le 4$ results.

\subsection{A bound for $n=3$}
For $n=3$ we have the bound $\qtot > 1/2$, which from theorem~\ref{thm:chargebound} requires that the exponents $m_i$ satisfy
\begin{align}
\label{eq:bound3}
\frac{1}{2} < \M{1}+\M{2}+\M{3} ~.
\end{align}
We will show that the solutions to this inequality for integers $m_1\ge m_2 \ge m_3 \ge 2$ fall into one of the following classes:
\begin{align}
(m_1,6,5) 		&\qquad 7\ge m_1 \ge 6~, \nonumber\\
(m_1,5,5)		&\qquad 9\ge m_1 \ge 5~, \nonumber\\
(m_1,7,4)		&\qquad 9\ge m_1 \ge 7~,\nonumber\\
(m_1,6,4)		&\qquad 11\ge m_1 \ge 6~,\nonumber\\
(m_1,5,4)		&\qquad 19\ge m_1 \ge 5~,\nonumber\\
(m_1,4,4)		& \qquad \infty >m_1\ge 4~,\nonumber\\
(m_1,11,3)	&\qquad 13\ge m_1\ge 11~,\nonumber\\
(m_1,10,3)	&\qquad 14\ge m_1\ge 10~,\nonumber\\
(m_1,9,3)		&\qquad 17\ge m_1\ge 9~,\nonumber\\
(m_1,8,3)		&\qquad 23\ge m_1\ge 8~,\nonumber\\
(m_1,7,3)		&\qquad 41\ge m_1\ge 7~,\nonumber\\
(m_1,m_2,3)	& \qquad \infty >m_1\ge m_2,~m_2 \le 6~,\nonumber\\
(m_1,m_2,2)	&\qquad \infty >m_1\ge m_2 \ge 2~.
\end{align}
There are $5$ infinite families that depend on one free integer, $1$ infinite family that depends on two free integers, and $99$ sporadic solutions.

To demonstrate this, we proceed as follows.
\begin{enumerate}
\item $m_3  < 6$.  This follows from
\begin{align}
 \frac{1}{2} <  \M{1}+\M{2}+\M{3} \le \frac{3}{m_3}~.
\end{align}
\item $m_1 \ge m_2 \ge m_3 = 5$.  This can hold only if
\begin{align}
\frac{3}{10} < \M{1}+\M{2} \le \frac{2}{m_2}~,
\end{align}
i.e. only for $m_2 = 5$ or $m_2 = 6$.
\begin{enumerate}
\item  $m_1 \ge m_2 = 6$ requires $m_1 = 6$ or $m_1 = 7$~;
\item $m_1 \ge m_2 = 5$ requires $m_1 \in \{5,6,7,8,9\}$~.
\end{enumerate}
\item $m_1\ge m_2 \ge m_3 = 4$.  This can hold only if
\begin{align}
\frac{1}{4} < \M{1} + \M{2} \le \frac{2}{m_2}~,
\end{align}
i.e. only if $m_2 \in \{4,5,6,7\}$~.
\begin{enumerate}
\item $(m_1,7,4)$ satisfies bound for $m_1 \le 9$.
\item $(m_1,6,4)$ satisfies bound for $m_1 \le 11$.
\item $(m_1,5,4)$ satisfies bound for $m_1 \le 19$.
\item $(m_1,4,4)$ satisfies bound for every $m_1$.
\end{enumerate}

\item $m_1 \ge m_2 \ge m_3 =3$.
The bound is satisfied only if
\begin{align}
\frac{1}{6} <\M{1}+\M{2} \le  \frac{2}{m_2}~.
\end{align}
This is possible only for $m_2 \in\{3,\ldots,11\}$.
\begin{enumerate}
\item $(m_1,11,3)$ satisfies bound for $m_1\le 13$.
\item $(m_1,10,3)$ satisfies bound for $m_1 \le14$.
\item $(m_1,9,3)$ satisfies bound for $m_1 \le17$.
\item $(m_1,8,3)$ satisfies bound for $m_1 \le 23$.
\item $(m_1,7,3)$ satisfies bound for $m_1 \le 41$.
\item $(m_1,m_2,3)$ with $m_2 \le 6$ satisfies bound for every $m_1$  .
\end{enumerate}

\item $m_1 \ge m_2 \ge m_3 =2$.  The bound is satisfied for every $m_1,m_2$.

\end{enumerate}

\subsection{A bound for $n=4$}
With $n=4$ variables the bound of interest is
\begin{align}
1 < \M{1}+\M{2}+\M{3} +\M{4} ~.
\end{align}
We proceed in the same fashion as in the previous section and establish that integer solutions $m_1\ge m_2\ge m_3 \ge m_4 \ge 2$ to this inequality fall into one of the following classes:
\begin{align}
\label{eq:bound4}
(5,4,4,3)~,\nonumber\\
(4,4,4,3)~, \nonumber\\
(m_1,5,3,3)~, &\qquad 7\ge m_1 \ge 5~,\nonumber\\
(m_1,4,3,3)~,&\qquad 11 \ge m_1 \ge 4~,\nonumber\\
(m_1,3,3,3)~,&\qquad \infty > m_1\ge 3~,\nonumber\\
(m_1,m_2,m_3,2)~,&\qquad \M{1} + \M{2}+\M{3} > \frac{1}{2}~.
\end{align}
Thus, we have all the solutions to the previous problem with $m_4=2$,  one additional infinite family, and $13$ sporadic solutions.

\begin{enumerate}
\item $m_4 <4$ follows from
\begin{align}
1 < \M{1}+\M{2}+\M{3}+\M{4} \le \frac{4}{m_4}~.
\end{align}
\item $m_1 \ge m_2 \ge m_3 \ge m_4 = 3$ is satisfied only if
\begin{align}
\frac{2}{3} < \M{1} + \M{2} + \M{3} \le \frac{3}{m_3}~.
\end{align}
This requires $m_3 \le 4$.
\begin{enumerate}
\item $m_1 \ge m_2 \ge m_3=4$ is satisfied only if
\begin{align}
\frac{5}{12} < \M{1}+\M{2},
\end{align}
and the only solutions are $m_1 = 5$, $m_2 = 4$ or $m_1 = m_2 = 4$~.
\item $m_1 \ge m_2 \ge m_3 = 3$ is satisfied only if
\begin{align}
\frac{1}{3} < \M{1} + \M{2} \le \frac{2}{m_2}~.
\end{align}
This requires $m_2 \le 5$ and has solutions 
\begin{align}
7\ge m_1\ge m_2 & = 5,~&
11\ge m_1 \ge m_2 & = 4~,&
\infty > m_1 \ge m_2 & = 3~.
\end{align}
\item $m_1 \ge m_2 \ge m_3 \ge m_4 = 2$ is exactly the $n=3$ bound~(\ref{eq:bound3})~.
\end{enumerate}

\end{enumerate}

\subsection*{A bound for $n=5$}
Our last bound is
\begin{align}
\frac{3}{2} < \M{1}+\M{2}+\M{3} +\M{4} +\M{5}~.
\end{align}
The integer solutions with $m_1\ge m_2 \ge m_3 \ge m_4 \ge m_5 \ge 2$ are
\begin{align}
\label{eq:bound5}
(5,3,3,3,3)~,\nonumber\\
(4,3,3,3,3)~,\nonumber\\
(3,3,3,3,3)~, \nonumber\\
(m_1,m_2,m_3,m_4,2)~, &\qquad  1 < \M{1}+\M{2}+\M{3}+\M{4}~.
\end{align}
There are no additional infinite families and just three more sporadic cases.

This is obtained by following the same steps as in the previous two cases.  Clearly $m_5 < 4$, and for $m_5 =3 $ we obtain only the two listed cases.  For $m_5 =2$ we reduce the problem to the $n=4$ case.

We can now give a complete description of the infinite families for $n=5$:
\begin{align}
(m_1,3,3,3,2)~, &&
(m_1,4,4,2,2)~, &&
(m_1,6,3,2,2)~, &&
(m_1,5,3,2,2)~, \nonumber\\
(m_1,4,3,2,2)~, &&
(m_1,3,3,2,2)~, &&
(m_1,m_2,2,2,2)~.
\end{align}
This matches the $7$ infinite families that appear in Table 3 of~\cite{Xie:2015rpa}~.  The remaining possibilities constitute a finite list with $115$ entries.

\bibliographystyle{./utphys}
\bibliography{./bigref}

\end{document}